\documentclass[a4 paper, 12pt,twoside]{article}
\usepackage{fullpage}                           % Makes the page margins smaller to a predefined one.
\usepackage[lmargin=1.0in,rmargin=0.8in,tmargin=0.4in,headsep=.2in]{geometry}
\usepackage{makecell}
\usepackage{multirow}
\usepackage{booktabs,tabularx}
\usepackage{stfloats}
\usepackage{dsfont}

\usepackage{graphicx}
\usepackage{forloop}
\usepackage[ruled,vlined]{algorithm2e}
\usepackage{float}
\usepackage{hyperref}
\usepackage{etoolbox}
\usepackage{dsfont}
\usepackage{url}
\usepackage{subfig}
\usepackage{amsmath,amsfonts,amsthm,bm}
\usepackage{float}
\usepackage[natbibapa]{apacite}
\usepackage{xtab} % To produce the empty space on the left
\usepackage[normalem]{ulem}
\usepackage{sectsty}% http://ctan.org/pkg/sectsty
\makeatletter
\def\@seccntformat#1{\@ifundefined{#1@cntformat}%
{\csname the#1\endcsname\;}%  default
{\csname #1@cntformat\endcsname}% individual control
}
\def\section@cntformat{\thesection.\;} % Dot after the section number
\def\subsection@cntformat{\thesubsection.\;} % Dot after the subsection number
\makeatother
\usepackage{hyperref}  
\hypersetup{
    colorlinks=true,      
    urlcolor=blue,
    linkcolor=black,
    citecolor = black}
 
\usepackage{fancyhdr}
\fancypagestyle{first}{%
  \fancyhf{}% clear all header and footer fields
  \fancyfoot[L]{{\footnotesize \textsf{: \href{}{\color{blue}\uline{}} $\;$}}}%
  \fancyhead[R]{{\bf\small \textsf{}}\\
\href{}{\color{blue}\uline{\textsf{}}}\\
\textsf{}\\
\textsf{}}%
}

\usepackage{amsmath,amsfonts,amssymb,amsthm,epsfig,epstopdf,url,array}
 \usepackage{makeidx,epsfig,lscape}
 \usepackage{xcolor,pict2e}
\usepackage{upgreek}
\theoremstyle{definition}

\newtheorem{theorem}{Theorem}
\newtheorem{remark}{Remark}
\newtheorem{definition}{Definition}
\newtheorem{proposition}{Proposition}
\newtheorem{lemma}{Lemma}

\numberwithin{theorem}{section}
\numberwithin{definition}{section}
\numberwithin{condition}{section}
\numberwithin{proposition}{section}
\numberwithin{lemma}{section}
\numberwithin{condition}{section}
\numberwithin{corollary}{section}
\numberwithin{equation}{section}
\begin{document}
\thispagestyle{first}
\vspace*{3cm}
%%%%%%%%%  TITLE %%%%%%%%%%%%%%%%%
{\noindent\huge  %Valuation of 
Hedging Crop Yields Against Weather Uncertainties--A Weather Derivative Perspective}\\[1cm]
%%%%%%%%%%%%%%%%  Author Data %%%%%%%%%%%%%%%%%%%
	{ \textbf
		{Samuel Asante Gyamerah$^\mathrm{}$\footnote{\emph{*Corresponding author: E-mail: saasgyam@gmail.com}}, Philip Ngare$ ^\mathrm{2}$,Dennis Ikpe$ ^{\mathrm{3},\mathrm{4}}$}}\\[1mm]
	$^\mathrm{1}${\footnotesize \it Pan African University, Institute of Basic Sciences Technology and Innovation, Kenya.
}
	$^\mathrm{2}${\footnotesize \it University of Nairobi, Kenya}	
	$^\mathrm{3}${\footnotesize \it African Institute for Mathematical Sciences, South Africa.}
	$^\mathrm{4}${\footnotesize \it Michigan State University, USA.}	
%%%%%%%%%%%   The Information Bar on the Left %%%%%%%%%%%

%%%%%%%%%%%%%%%  The Abstract and Keywords %%%%%%%%%%%%%%%%%%%%%%%%
{\rule{0.7\textwidth}{2pt}}\\[0.2cm]
{\bf\large Abstract}\\
The effects of weather on agriculture in recent years have become a major global concern. Hence, the need for an effective weather risk management tool (i.e., weather derivatives) that can hedge crop yields against weather uncertainties. However, most smallholder farmers and agricultural stakeholders are unwilling to pay for the price of weather derivatives (WD) because of the presence of basis risks (product-design and geographical) in the pricing models. 
To eliminate product-design basis risks, a machine learning ensemble technique was used to determine the relationship between maize yield and weather variables. The results revealed that the most significant weather variable that affected the yield of maize was average temperature. A mean-reverting model with a time-varying speed of mean reversion, seasonal mean, and local volatility that depended on the local average temperature was then proposed. The model was extended to a multi-dimensional model for different but correlated locations. Based on these average temperature models, pricing models for futures, options on futures, and basket futures for cumulative average temperature and growing degree-days are presented. Pricing futures on baskets reduces geographical basis risk, as buyers have the opportunity to select the most appropriate weather stations with their desired weight preference. With these pricing models, farmers and agricultural stakeholders can hedge their crops against the perils of extreme weather.
\vspace{0.5cm}\\
{\bf\large Keywords}\\
basis risk, agricultural risk management, weather derivatives, cumulative average temperature, growing degree-days 
\vspace{0cm}\\
{\rule{0.7\textwidth}{2pt}}
%%%%%%%%%%%%%%%%%  The Document Starts Here %%%%%%%%%%%%%%
\begin{itemize}
	\item Published as:\\
	Samuel Asante Gyamerah, Philip Ngare and Dennis Ikpe.\\
	Mathematical and Computational Applications, Vol. 24, No. 3, 2019\\
	\url{https://doi.org/10.3390/mca24030071} 
\end{itemize}

\section{Introduction}
\label{S:1}
Agriculture continues to be an important sector that contributes to Ghana's exports earnings, inputs for most manufacturing sectors, and revenue generation for majority of the population. The agriculture and agribusiness sector account for a significant share of the major economic activities in Ghana and a major source of income for most smallholder farmers. It is reported to account for about 70\% of the labour force and more than 25\% of the gross domestic product in Africa \citep{uneca2009}. This makes agriculture one of the most important and largest sector in the development of the economies in Africa, of which Ghana is a member state. However, in Ghana, the sector continues to be controlled by primary production as a result of high weather variability and hydrological flows, especially in the Northern savannas \citep{ibn2018nexus}. The Northern savannas of Ghana have experienced perennial extreme flooding and droughts, both linked to extreme heat and temperatures \citep{ibn2018nexus}. These have contributed to crop failures and have consequently led to extensive impacts on the economic activities of most rural farmers.
\\
[2mm]
Weather variables are difficult to mitigate, especially for smallholder farmers in most developing and under-developed countries, and have great effects on the farming activities of these farmers. For this reason, an effective and reliable risk management tool, weather derivative (WD), is needed to hedge farmers and stakeholders from the peril of weather uncertainties. By linking the payoffs to a fairly measured weather index (e.g. temperature, rainfall, humidity, and sunshine), a WD reduces or eliminates the disadvantages of traditional insurance, such as moral hazards and information asymmetry. Turvey \cite{turvey2001weather} examined the pricing of weather derivatives in Ontario and contended that weather derivatives and weather insurance can be used as a form of agricultural risk management tools. Zong and Ender  \cite{zong2016spatially} developed a novel type of weather derivatives contract called a climatic zone-based growth degree-day contract. Their aim was to mitigate weather risk in the agricultural sector of mainland China by introducing new types of temperature indices. Even though several weather risk management tools have been recently introduced into the WD market for smallholder farmers in most countries around the world, their purchases have been lower than expected. Among the reasons causing the low purchase of WD and the unwillingness of farmers to pay for WD in most communities are the lack of capacity to design and determine the value of this insurance product and, most importantly, high basis risk (product-design and geographical) in the contract design and implementation \cite{woodard2008basis}. Basis risk in WD market is defined as the difference between the actual loss and the actual payout of WD \cite{rohrer2004relevance}. Using a different index rather than the actual index that affects a specific crop yield at a location can lead to a larger gap between the real exposure and the payoff (product-design basis risk). Product design basis risk can be mitigated if the appropriate weather observation is used to construct the index for WD. Hence, a complete evaluation of the relationship between historical weather and crop yield data is significant for an effective and reliable design of WD. Forecasting of crop yields and feature importance of different weather indices for crop yields are, therefore, principal components for an effective rate-making process in the insurance/derivatives market. Another form of basis risk--geographical basis risk occurs when there is a deviation of weather conditions at the measurement station of the weather derivative and the weather conditions at the location of the buyer \cite{ritter2014minimizing}. Spicka and Hnilica \cite{spicka2013methodical} evaluated the effectiveness of weather derivatives as a revenue risk management tool by {considering} crop growing conditions in the Czech Republic. The authors concluded that high basis risk can significantly misrepresent the payoff of the contract.  Musshoff et al. \cite{musshoff2011management}, in their research, concluded that hedging effectiveness for the agricultural sector using weather derivatives is controlled by the contract design. They categorized basis risk into local and geographical basis risk, and asserted that basis risk has a greater influence on the hedging effectiveness of weather derivatives.
\\
In this study, geographical basis {risk} is mitigated by pricing futures on a temperature basket rather than a single index {contract}. This requires the determination of the correlation between the locations under study. 
\\
[2mm]
In the literature of weather derivative pricing, different methods have been proposed for pricing temperature-based weather derivatives. Among these methods are the indifference pricing approach, actuarial pricing methods, the equilibrium model approach, and incomplete market pricing models. The indifference pricing approach is a valuation method for weather derivatives which is founded on the arguments of expected utility. It is also referred to as the incomplete market pricing model, reservation price, or private valuation. This approach for pricing has been used in most traditional financial derivatives where the market is incomplete (see \cite{zariphopoulou2001solution, henderson2002valuation,owen2002utility}). In the weather derivative market, Davis \cite{davis2001pricing} used the marginal utility approach or "shadow price" method of mathematical economics to price weather derivatives. The basis for their pricing approach was centered on the concept that investors in the weather derivative market are not representative but experience distinct risks that are linked to the effect of weather on their business. Brockett et al. \cite{brockett2009pricing} used the indifference pricing approach to value weather derivative futures and options. The relationships between the indifference pricing and actuarial pricing approaches for weather derivatives were studied in a mean-variance context, where they provided instances that the actuarial pricing method does not give a distinct valuation for weather derivatives. From the concept of utility maximization, Barrieu and Karoui \cite{barrieu2002optimal} calculated the optimal profile (and its value) of derivatives written on an illiquid asset; for example, the weather or a catastrophic event. Cao and Wei \cite{cao2000equilibrium} generalized the equilibrium model of \cite{lucas1978asset} by incorporating weather as a basic variable in the economy. Based on this, they proposed an equilibrium valuation framework for temperature-based weather derivatives. Using a model that captured the daily temperature dynamics of five cities in the United States of America (Atlanta, Chicago, Dallas, New York, and Philadelphia), they performed numerical analysis for forward and option contracts on heating degree-days (HDDs) and cooling degree-days (CDDs). Their analysis revealed that the market price of risk is mostly trivial when linked to the temperature variable, particularly when the aggregate dividend process is mean-reverting. They showed that unrealistic assumptions in historical simulation methods result in inaccurate pricing of temperature-based weather derivatives. As stated by \cite{alaton2002modelling}, the weather derivative market is a typical example of an incomplete market. Hence, pricing models based on incomplete markets are the most applicable valuation method for weather derivatives. These models consider the hedgeable and unhedgeable components of risk. The price of a weather derivative is usually dependent on different weather indices, such as HDD, CDD, Pacific Rim (PRIM), cumulative average temperature (CAT), and growing degree-days (GDD). Different authors \citep{alaton2002modelling, mraoua2007temperature} have used the HDD and CDD indices as the major indices for pricing weather derivatives for the energy industry.  
\\
[2mm]
The contributions made in this study are: (1) We are able to empirically determine the main underlying weather variable (average temperature) that affects the yield of the selected crop using machine learning ensemble techniques and feature selection through crop yield forecasting, rather than the usual assumption of using temperature as the underlying without proper empirical studies. This will eliminate product-design basis risk during pricing of the weather derivatives. (2) Previous studies have either used a piecewise constant volatility function or a seasonal volatility (e.g.,~\cite{alaton2002modelling, benth2005stochastic, benth2007volatility}). However, our proposed model includes a local volatility which is able to capture the local variations of the daily average temperature. (3) Futures, options on futures, and basket futures {({futures for multi-dimensional locations})} on CAT and GDD are priced using the constructed daily average temperature model. These pricing models for CAT and GDD are the first of their kind in the literature. (4) The basket futures pricing will help in mitigating geographical basis risks in the weather derivative market.

\section{Crop Yield--Weather Model and Feature Importance for Weather Derivatives}
Different factors, such as condition of the soil, the variety of seedling, and type and amount of fertilizer applied to the soil affect the yield and productivity of crops. These factors can be controlled by the farmer and industry players. Weather variables, especially surface temperature, rainfall, and humidity, are the principal drivers of the differences in crop yields \citep{hu2003climate}. These weather variables directly affect the moisture content of the soil and the level of nutrients in the soil. The effect of the uncertainties in the pattern of weather, both between and within planting seasons, can affect the crop production and yield significantly. This has significantly affected the yield of most crops, causing economic and food security risks in most developing and under-developed countries. Maize is the most widely cultivated staple crop in the Northern savanna, and is the major source of income for about 45\% of households there \citep{wood2013agricultural}. For this reason, the yield of maize was used as a proxy to determine the effect of the selected weather variables on crop yield. This weather variable can, then, be used as the underlying for weather derivatives pricing.

\subsection{Machine Learning Ensemble Technique for Weather-Crop Yield Model}
\label{S:2}
Machine learning (ML) is used in regression and classification to solve most of the problems that arise as a result of nonlinearity in the features and the response variable. In this study, ML ensemble classification algorithms are used to predict the possibility of an improved crop yield harvest or crop yield loss (a two-class (binary) problem). The choice of an optimal ML algorithm for prediction is a major factor to consider in any forecasting problem. In our case, the chosen ML technique should be able to predict whether there will be an increase or decrease in crop yields for a period of years with a small margin of error. The target variable for the ensemble classifier is whether there will be a loss or increase in crop yields for the year-ahead harvest; that is,
\begin{itemize}
	\item If the long-term average crop yield $(\bar{Y})$ is smaller than or equal to the present years crop yield $(Y_{t})$, then there is an increase in crop yield; else:
	\item If the long-term average crop yield $(\bar{Y})$ is greater than the present years crop yield $(Y_{t})$, then there is a decrease or loss in crop yield.
\end{itemize}
We assign labels to this classification: ``0" for a decrease/loss in crop yields and ``1" for an increase in crop yields, 
\begin{equation}
y=
\begin{cases}
1 \qquad\qquad \mbox{if} \quad \bar{Y} \le Y_t,
\\
0 \qquad\qquad \mbox{if} \quad \bar{Y} > Y_t.
\end{cases}
\label{classification}
\end{equation}  
\vspace{6pt}
Humidity, sunlight, rainfall, minimum temperature (minT), maximum temperature (maxT), and average temperature (aveT) are the features ({\bf X}) used in predicting the target variable ({\bf Y}). Now, suppose a set of features ${\bf X}=\{x_i \in \mathbb{R}^n\}$ with associated labels ${\bf Y}=\{y_i \in Y, y_i \in (0,1) \}$, are given for a training data set $T=\{(x_i, y_i)\}$. In this way, we solve the supervised classification problem where the learning model $N$ depends on $D$. 
\\[2mm]
The yearly datasets (2000--2016) for maize yield Bole, Tamale, and Yendi were taken from the Statistics, Research, and Information Directorate (SRID) of the Ministry of Food and Agriculture, Ghana. Daily historical data for sunlight, humidity, rainfall, maximum and minimum temperature from 2000--2016 were obtained from the Ghana Meteorological Agency, Ghana. K-nearest neighbors (KNN) algorithm was used to compute the missing data points in the dataset. The yearly data points were calculated from the daily data of the selected weather variables using the arithmetic average. The yearly average temperature was computed from the arithmetic average of the yearly maximum and minimum temperature. Data for the weather variables was also taken from Bole, Tamale, and Yendi. These towns are part of the Northern Savanna and they are considered to be the food basket of Ghana. Due to the sensitivity of the ML algorithms used and the unequal weight of the data sets, the sample data sets were set into an identical scale--the min-max normalization scale in the interval $[0,1]$. The sample data sets were divided into training ($80\%$) and testing ($20\%$) data sets. The training data set was used to build the classification ensemble algorithm and the testing data was used to validate the constructed ensemble model. The "accuracy"{({{accuracy is an evaluation metric  used in classification and regression problems}})} of the training model is fined tuned using grid search. The training model was used to predict the crop yield signal (an increase crop yield or a decrease in crop yields). Using accuracy and the receiver operating characteristics curve (ROC) or area under the curve (AUC), the predictions on the testing data set were evaluated. 

\subsection{Model Evaluation and Feature Importance}
Stacking is an ensemble learning technique that combines multiple classification or regression models through a meta-classifier/meta-regressor to a single predictive model to reduce bias (boosting), variance (bagging), and improve the accuracy of predictions (stacking). Using the training data set, the base level classifiers were trained. The meta-model was trained on the outputs of the base level algorithm as features. Stacking ensemble learning algorithms are heterogenous because the base level classifiers are made of different classification algorithms. In this study, Adapted boosting (AdaBoost) and artificial neural network (ANN), were used as the base classifier and gradient boosting machine (GBM) was used as the meta-classifier. Stacking ensemble algorithm is outlined in Algorithm \ref{stacking_algorithm},   
\newline
\begin{algorithm}[H]
	\KwIn{$D : \{(x_i,y_i)\} \mid x_i \in X, y_i \in Y$}
	\KwOut{An ensemble classifier H}
	\nl {\it Step 1.}\: {\it Learn base-level classifiers} \\
	\nl  {\bf for} $t \leftarrow 1$ to $N$ {\bf do} \\
	\nl  \qquad Learn a base classifier $h_t$ based on $T$, $h_t = H_t(T)$ \\
	\nl {\bf end for}\\
	\nl {\it Step 2.}\: {\it Construct new data set from $T$, $T' = \phi$} \\ 
	\nl {\bf for} $i \leftarrow 1$ to $n$ {\bf do} \\ 
	\nl \qquad Construct a new data set that contains  $\{x_i^{new}, y_i\}$ \\ where $\{x_i^{new}, y_i\} = \{h_t(x_i)\,\,\mbox{for}\,\,t=1 \,\,\mbox{to} \,\, N\}$ \\
	\nl {\bf end for}\\
	\nl {\it Step 3.}\: {\it Learn a second-level classifier} \\
	\nl Learn a new classifier $h^{new}$ based on the newly constructed data set.\\
	\nl {\bf Return} $H(x) = h^{new}(h_1(x), h_2(x), \cdots, h_N(x)$)
	\caption{\bf Algorithm of Stacking Ensemble} \label{stacking_algorithm}
\end{algorithm}
\vspace{6pt}
The performance of the evaluation metrics of the proposed stacking ensemble classifier is presented in Table \ref{evaluation_metrics}. The table gives the performance of the ensemble classifier on the testing data set. For binary classification, an AUC value of $50\%$ or less is as good as randomly selecting the labels. An AUC value closer to 1 and an accuracy value closer to $100\%$ indicate the superiority of the proposed model. In general, the classification model was very optimal in predicting the class of the target variable (crop yield) using the selected features (weather variables).
\begin{table}[H] 
	\centering
	\begin{tabular}{cccccc}
		\toprule 
		& {\bf Bole}  & {\bf Tamale} & {\bf Yendi} 
		\\ 
		\midrule
		{\bf Accuracy} & 0.8319 & 0.8207 & 0.8895  
		\\ 
		\midrule
		{\bf AUC} & 0.8104 & 0.7991 & 0.8334  
		\\ 
		\bottomrule 
	\end{tabular}
	\caption{Performance of evaluation metrics of the proposed stacking ensemble classification.}
	\label{evaluation_metrics}
\end{table}
\noindent
From the feature importance in Table \ref{feature_importance}, average temperature, rainfall, and maximum temperature were the three most important features in the classification model used in predicting the yield of maize at Bole. In Tamale, the three most important variables that contributed to predicting the yield of maize were average temperature, minimum temperature, and rainfall. With importance values of $0.3017$, $0.3006$,  and $0.2875$, rainfall, average temperature, and maximum temperature were the three variables that had most effect in predicting maize yield at Yendi. Generally, from these results, average temperature can be considered as the most important underlying weather variable that affects the yield of maize in the Northern region of Ghana.

\begin{table}[H]
	\renewcommand{\arraystretch}{1.0}
	\centering	
	\resizebox{\columnwidth}{!}{%
		\begin{tabular}{>{\bfseries}c*{9}{c}} 
			\toprule
			\multirow{2}{*}{\bfseries Feature}
			& 
			\multicolumn{2}{c}
			{\bfseries  Bole}
			&
			\multicolumn{2}{c}
			{\bfseries \qquad\qquad Tamale}
			&
			\multicolumn{4}{c}
			{\bfseries \qquad\qquad Yendi}
			\\
			\cmidrule(lr){2-3} 
			\cmidrule(lr){5-6}
			\cmidrule(lr){8-9}
			& \thead{\bfseries{Importance} \\ \bfseries{Value}} & \textbf{Rank} 
			& 
			& \thead{\bfseries{Importance} \\ \bfseries{Value}} & \textbf{Rank} 
			&
			& \thead{\bfseries{Importance} \\ \bfseries{Value}} & \textbf{Rank} 
			\\ \hline
			minT & 0.0908 & 4 && 0.2001 & 2 && 0.2452 & 4      
			\\ 
			maxT & 0.2017 & 3 && 0.1402 & 4 && 0.2875 & 3           
			\\
			aveT & 0.3042 & 1 && 0.2075 & 1 && 0.3006 & 2             
			\\
			Rainfall & 0.2552 & 2 && 0.1905 & 3 && 0.3017 & 1          
			\\
			Sunlight  & 0.0646 & 5 && 0.0625 & 6 && 0.2501 & 5         
			\\
			Humidity & 0.0579 & 6 && 0.0983 & 5  && 0.2466 & 6    
			\\
			\bottomrule
		\end{tabular}
	}
	\caption{Performance of evaluation metrics of the proposed stacking ensemble classification.}
	\label{feature_importance}
\end{table}

\section{Temperature-Based Weather Derivatives}
Even though the literature on weather derivatives has evolved rapidly over the past few decades, a consensus theoretical framework for evaluating weather derivatives has not been reached. This can be attributed to the fact that the  underlying factors of WD are not tradable in the financial market and, hence, traditional pricing approaches cannot be used for the valuation of this product. Further, weather indices do not correlate strongly with the prices of other financial products in the financial market and, as a result, the underlying indices cannot be substituted for a linked exchange security.

\subsection{Previous Temperature Dynamics Models} 
Alaton et al. {\cite{alaton2002modelling}  proposed the following mean-reverting model for the dynamics of temperature variations:
	\begin{equation}
	dT(t) = dS(t) + \beta(T(t) - S(t))dt + \sigma(t)dB(t),
	\label{alaton}
	\end{equation}
	where $T(t)$ represents the daily average temperature, $S(t)$ is the deterministic seasonal component, $B(t)$ is a Brownian motion, $\beta$ is a constant mean-reversion rate, and $\sigma(t)$ is the volatility. They assumed that the volatility is a piecewise constant function which characterizes the monthly variation in volatility. From the proposed model of \cite{alaton2002modelling}, Benth and Benth \cite{benth2005stochastic} suggested the following mean-reverting model for the time-dynamics of Norwegian daily average temperature:
	\begin{equation}
	dT(t) = dS(t) + \beta(T(t) - S(t))dt + \sigma(t)dL(t).
	\label{benth_1}
	\end{equation}
	As indicated by \cite{benth2005stochastic}, the only innovation from Equation (\ref{alaton}) is the introduction of a L\'{e}vy noise, $L(t)$, instead of the Brownian motion $B(t)$. They used the generalized hyperbolic distribution, which is a special class of L\'{e}vy process, to capture the skewness and (semi-) heavy tails of their temperature data. However, to allow analytical pricing using the temperature dynamics model, Benth and Benth~\cite{benth2007volatility} later proposed to use a Brownian motion, instead of the L\'{e}vy process, in Equation (\ref{benth_1}). Clearly, References  \cite{alaton2002modelling, benth2005stochastic, benth2007volatility} used a constant mean-reversion rate, a volatility which is piecewise constant function, and they did not consider multi-dimensional locations in their models.    
	
	The introduction of a time-varying speed of mean-reversion, a local volatility that is able to capture the local variations of the daily average temperature at the selected locations, and a multi-dimensional daily average temperature model for multi-locations gave rise to the innovations in our proposed model and pricing. Our proposed models capture all the stylized facts of the selected locations and, to the best of our knowledge, it is the first of its kind in the literature.}

\subsection{Daily Average Temperature Data}

The daily average temperature data (degree Celsius) for Bole and Tamale in the Northern savanna of Ghana, over a measurement period from 01/01/1992 to 31/08/2017, were taken from the Ghana Meteorological Service. The data consisted of the daily maximum and minimum temperatures. Missing data points were computed using KNN. The daily average temperature was computed from the arithmetic average of the daily minimum and maximum temperatures. For consistency in days (365 days) per year, February 29 was removed from the dataset for each leap year. As a result, the total daily average temperature had 9368 data points. 

The seasonal and seasonally-adjusted (de-seasonalized) plot of the daily average temperature data for Bole and Tamale are presented in Figures \ref{seasonal_deseasonal_bole} and \ref{seasonal_deseasonal_tamale}, respectively. 
The spatial correlation between Bole and Tamale was estimated, in order to develop a basket temperature for weather derivatives. The de-seasonalized daily average temperature (as suggested by \cite{alexandridis2012weather}) was used to estimate the correlation matrix between the two towns (Bole and Tamale). From Table \ref{correlation_matrix}, it is clear that the average temperature correlation between the selected locations was very high. This indicates that a maize farm at Bole will encounter the same weather-related risk as a maize farm in Tamale or Yendi, and vice versa. Hence, a weather station in any of these locations can be used for the contract without introducing geographical basis risk in the contract. A farmer with two or more farms at these locations can buy a single temperature basket derivatives contract.

\begin{figure}[H]
	\centering
	\includegraphics[height=9cm,width=16cm]{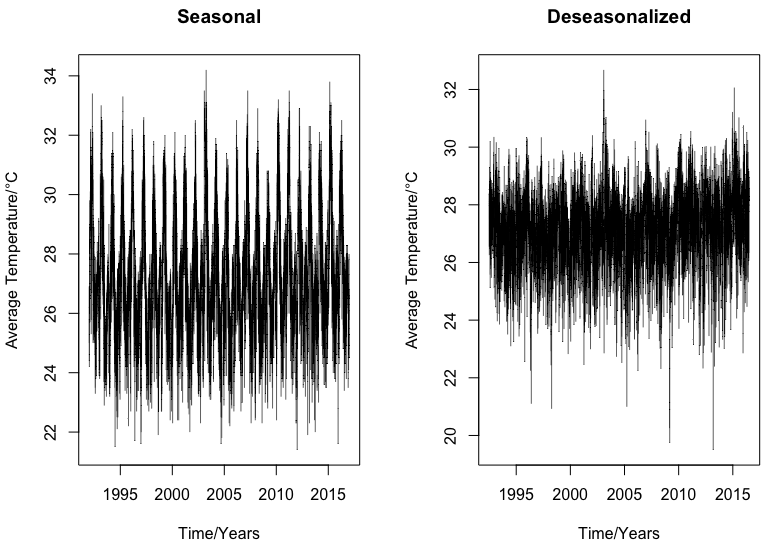}
	\caption{Seasonal and de-seasonalized daily average temperature of Bole.}
	\label{seasonal_deseasonal_bole}
\end{figure}

\begin{figure}[H]
	\centering
	\includegraphics[height=9cm,width=16cm]{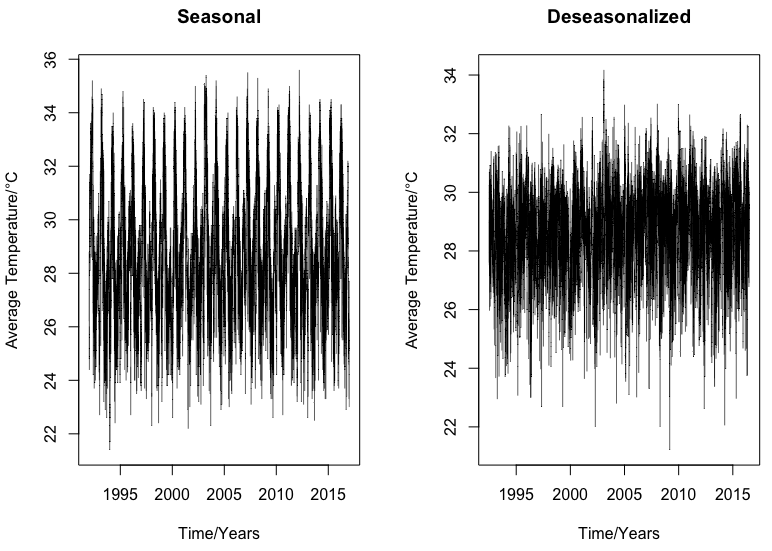}
	\caption{Seasonal and de-seasonalized daily average temperature of Tamale.}
	\label{seasonal_deseasonal_tamale}
\end{figure}

\begin{table}[H]
	\centering
	\begin{tabular}{cccc}		\toprule 
		& \textbf{Bole}  &  \textbf{Tamale} & \textbf{Yendi} \\ 
		\midrule
		\textbf{Bole} &   1  & 0.8733 & 0.8547
		\\
		\textbf{Tamale} & 0.8733   & 1 & 0.8998
		\\
		\textbf{Yendi} & 0.8547 & 0.8998 & 1
		\\
		\bottomrule 
	\end{tabular} 
	\caption{Correlation matrix of the de-seasonalized daily average temperature.}
	\label{correlation_matrix}
\end{table}

\subsection{Stochastic Dynamics of Daily Average Temperature}
Motivated by \cite{jones2003dynamics, alaton2002modelling, benth2005stochastic}, we propose a new average temperature dynamics model which is able to capture major stylized facts of daily average temperature, such as locality features, seasonality, mean-reversion, and volatility. These stylized facts were consistent with the daily average temperature of the chosen location for this study. The proposed one-dimensional model was extended to a multi-dimensional temperature model to cover the selected locations under study. For convenience and analytical tractability, we assumed that the residuals of the daily average temperature were independently and identically distributed (i.i.d.) standard normal. The proposed daily average temperature model is given as
\begin{equation}
dT(t) = dS(t) + \beta(t) \big(T(t) - S(t)\big)dt + \sigma T(t)dB(t), 
\label{proposed_temperature_model}
\end{equation}
where $T(t)$ represents the daily average temperature, $S(t)$ is the deterministic seasonal component, $\beta(t)$ is the time-varying speed of mean-reversion, and $\sigma T(t)$ is the daily average temperature volatility through~time.  

Following \cite{alaton2002modelling}, the seasonality component is defined as
\begin{equation}
S(t) = A + Bt+ C \sin \bigg(\frac{2 \pi t}{365} + \vartheta \bigg),
\label{seasonal_component}
\end{equation}
which is made up of a seasonal component $(C \sin {2 \pi t}/365 + \vartheta)$ and a trend component $(A + Bt)$, where $A$ and $B$ denote the constant and the coefficient of the linearity of the seasonal trend, respectively; $C$ denotes the daily average temperature amplitude, $\vartheta$; and $t$ is the time, measured in days.

Equation (\ref{seasonal_component}) can be transformed to
\begin{equation}
S(t) = a + bt + c \,\,\sin \bigg(\frac{2 \pi t}{365}\bigg) + d \,\, \cos \bigg(\frac{2 \pi t}{365}\bigg).
\label{transformed_seasonal_component}
\end{equation}
By comparing (\ref{seasonal_component}) to (\ref{transformed_seasonal_component}), the relationship of the parameters is given below
\begin{equation*}
A=a;\quad B=b;\quad C=\sqrt{c^2 + d^2}; \quad \vartheta = \arctan \bigg(\dfrac{d}{c}\bigg).
\end{equation*}
The numerical values of the constant in Equation (\ref{transformed_seasonal_component}) are estimated by fitting the function to the historical daily average temperature data using the method of least squares. The seasonal component for Bole and Tamale are given in the following function, respectively, 
\begin{equation}
\begin{aligned}
S(t) &= 22.15 + (4.57 \cdot 10^{-5})t + 1.98 \sin \bigg(\frac{2 \pi t}{365} - 67.71 \bigg),
\\
S(t) &= 18.38 + (7.03 \cdot 10^{-5})t + 2.06 \sin \bigg(\frac{2 \pi t}{365} - 72.89 \bigg).
\end{aligned}
\label{transformed_seasonal_component_values}
\end{equation} 
Using additive seasonal decomposition by moving averages, the daily average temperature was decomposed into a seasonal, linear, and a random (residual) component as shown in Figures \ref{dat_seasonal}--\ref{dat_residuals} respectively. The trend component (Figure \ref{dat_trend}) was smoother than the actual daily average temperature data plot (Figures \ref{seasonal_deseasonal_bole} and \ref{seasonal_deseasonal_tamale}, seasonalized) and captured the main movement of the daily average temperature data without the minor variations. The estimated seasonal component of the daily average temperature is presented in figure \ref{dat_est_seasonal_fig}. 

\begin{figure}[H]
	\centering
	\includegraphics[height=9cm,width=16cm]{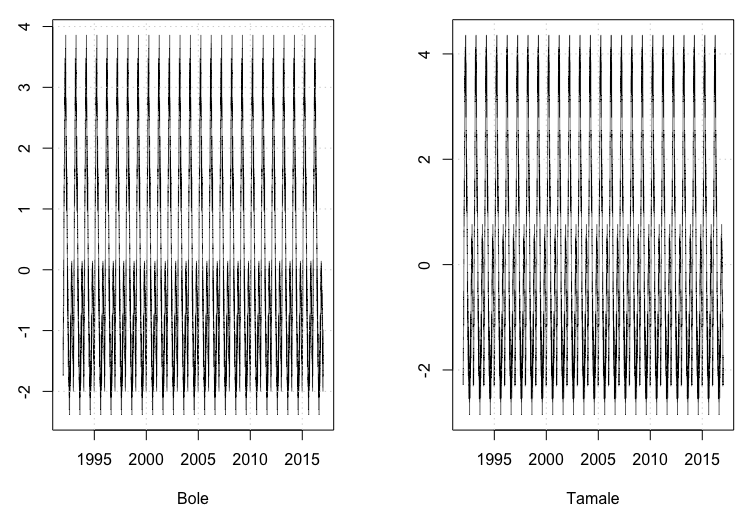}
	\caption{Seasonal trend of the daily average temperature for Bole and Tamale.}
	\label{dat_seasonal}
\end{figure}

\begin{figure}[H]
	\centering
	\includegraphics[height=11cm,width=16cm]{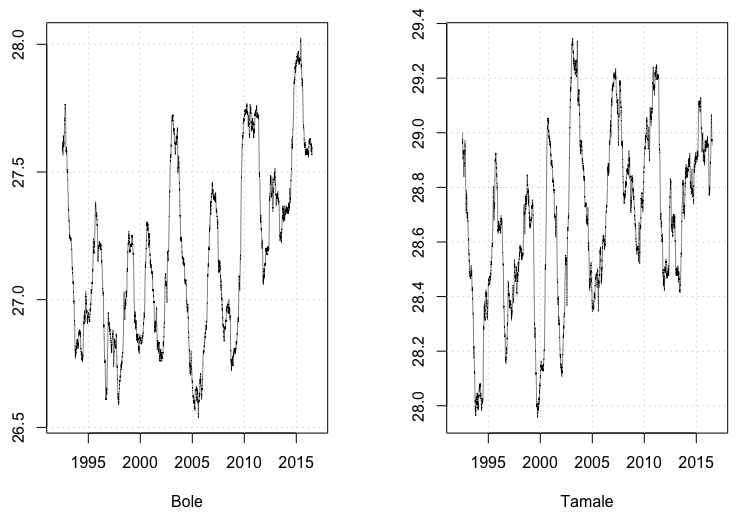}
	\caption{Linear trend of the daily average temperature for Bole and Tamale.}
	\label{dat_trend}
\end{figure}

\begin{figure}[H]
	\centering
	\includegraphics[height=11cm,width=16cm]{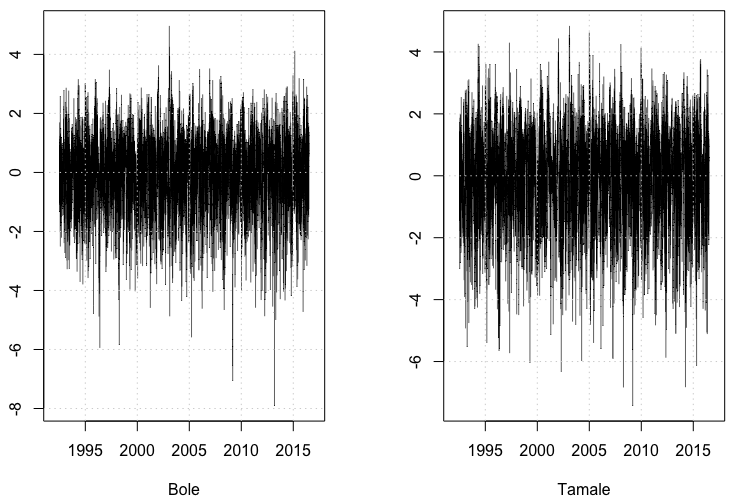}
	\caption{Residuals of the daily average temperature for Bole and Tamale.}
	\label{dat_residuals}
\end{figure}

\begin{figure}[H]
	\centering
	\includegraphics[height=9cm,width=16cm]{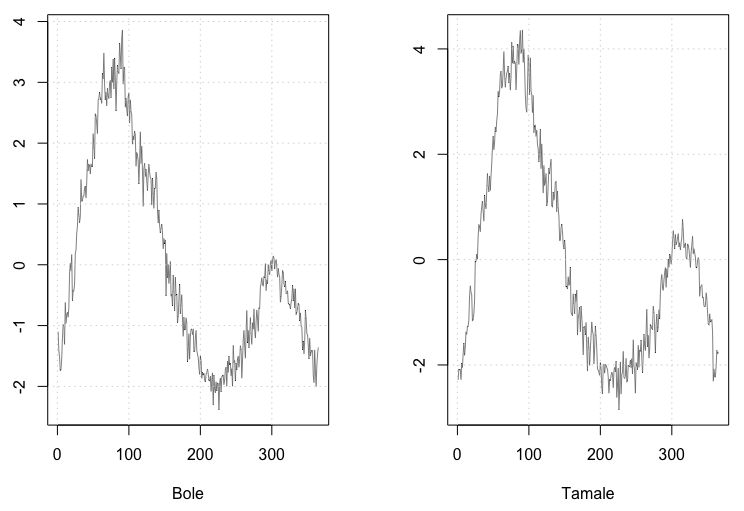}
	\caption{Estimated seasonal figure of the daily average temperature for Bole and Tamale.}
	\label{dat_est_seasonal_fig}
\end{figure}

\subsection{Temperature-Based Weather Derivative Pricing}
\begin{theorem}[Girsanov Theorem]
	Let $B_t$ be a Brownian motion on a probability space $(\Omega, \mathcal{F}, \mathbb{P})$ and $\lambda = \{\lambda_t : 0 \le t \le T \}$ be an adaptive process satisfying the Novikov condition
	\begin{equation}
	\mathbb{E} \left[  \exp\left(\frac{1}{2}  \int_{0}^{t} \lambda^2_u du            \right)       \right]      < \infty.
	\end{equation}
	\begin{equation*}
	\mbox{Let} \qquad Z(t) = \exp \left( \int_{0}^{t} \lambda_u dB_u - \frac{1}{2} \int_{0}^{t}  \lambda^2_u du  \right).
	\label{Novikov}
	\end{equation*}
	Then, $\mathbb{Q} \sim \mathbb{P}$ can be determined by the Radon--Nikodym derivative
	\begin{equation}
	\dfrac{d\mathbb{Q}}{d\mathbb{P}} \mid \mathcal{F}_t = Z(T).
	\end{equation}
	
	Then, the random process		
	\begin{equation}
	\begin{aligned}
	W_t &= B_t - \int_{0}^{t} \lambda_s ds,
	\\
	dW_t &= dB_t - \lambda_tdt,
	\label{girsanov}
	\end{aligned}
	\end{equation}	
	
	is a standard Brownian motion under the measure $\mathbb{Q}^{\lambda}$. 
	\label{theorem_girsanov}
\end{theorem}

\begin{remark}
	The Novikov condition in the Girsanov theorem ensures that $Z$ is positive martingale whenever $\mathbb{E}(Z) = 1$. This is referred to as the Radon--Nikodym derivative.
	\label{remark_novikov_condition}
\end{remark}

\begin{remark}\label{remark_MPR}
	$\lambda$ is refered to as the market price of risk (MPR). As there is no real weather derivative market in Africa from which the prices can be obtained, $\lambda$ is assumed to be a constant. For a constant $\lambda$, Equation (\ref{girsanov}) can be re-defined as
	\begin{equation}\label{girsanov_constant}
	dW_t = dB_t - \lambda dt.
	\end{equation}
\end{remark}  	
\begin{lemma}
	If the daily average temperature follows the proposed model in Equation (\ref{proposed_temperature_model}), then the explicit solution is given by 
	\begin{equation}
	T_t = S_t + (T_0 - S_0) e^{\int_{0}^{t}\beta_sds} + e^{\int_{0}^{t}\beta_sds} \int_{0}^{t} \sigma T_u e^{\int_{0}^{t}\beta_sds} dB_u.
	\label{lemma_explicit_proposed_temperature_model}
	\end{equation}
\end{lemma}

\begin{proof}
	We have that
	\begin{equation*}
	dT_t = dS_t + \beta_t (T_t - S_t)dt + \sigma T_tdB_t
	\end{equation*}
	%d(T_t - S_t) = 	\beta_t (T_t - S_t)dt + \sigma T_tdB_t 
	\begin{equation}
	d\tilde{T}_t = \beta_t \tilde{T}_t + \sigma T_tdB_t, 
	\label{proof_temperature_model_1}
	\end{equation}
	where $\tilde{T} = T_t - S_t$. Using the transformation below, $d\tilde{T}_t$ can be evaluated, 
	\begin{eqnarray*}
		F(\tilde{T}_t, t) = \tilde{T}_t e^{-\int_{0}^{t}\beta_sds} \end{eqnarray*}
	$$
	\dfrac{\partial F}{\partial \tilde{T}_t} = e^{-\int_{0}^{t}\beta_sds}; \qquad \dfrac{\partial^2 F}{\partial \tilde{T}_t^2} = 0; \qquad \dfrac{\partial F}{\partial t} = -\beta_t \tilde{T}_t e^{-\int_{0}^{t}\beta_sds}.
	$$
	
	Applying It\^{o}'s Lemma and Equation (\ref{proof_temperature_model_1}), 
	\begin{equation}
	dF_t = \sigma T_t e^{-\int_{0}^{t}\beta_sds} dB_t.
	\label{proof_temperature_model_2}
	\end{equation}
	
	Integrating Equation (\ref{proof_temperature_model_2}) over the interval $[0,t]$,
	\begin{equation*}
	\begin{aligned}
	F_t &= F_0 + \int_{0}^{t} \sigma T_u e^{-\int_{0}^{u}\beta_sds} dB_u \\
	\tilde{T}_t e^{-\int_{0}^{t}\beta_sds} &= \tilde{T}_0 + \int_{0}^{t} \sigma T_u e^{-\int_{0}^{u}\beta_sds} dB_u \\
	\tilde{T}_t &= \tilde{T}_0 e^{\int_{0}^{t}\beta_sds} + e^{\int_{0}^{t}\beta_sds} \int_{0}^{t} \sigma T_u e^{-\int_{0}^{u}\beta_sds} dB_u \\
	T_t &= S_t + (T_0 - S_0) e^{\int_{0}^{t}\beta_sds} + e^{\int_{0}^{t}\beta_sds} \int_{0}^{t} \sigma T_u e^{-\int_{0}^{u}\beta_sds} dB_u.
	\end{aligned}
	\end{equation*}
\end{proof}

\begin{lemma}\label{lemma_explicit_girsanov_temperature_model}
	Under the risk-neutral measure $\mathbb{Q}$, the explicit solution of the daily average temperature model is given by
	\begin{equation}
	T_t = S_t + (T_0 -S_0) e^{\int_{0}^{t}\beta_sds} +
	\int_{0}^{t} \sigma \lambda T_u e^{\int_{u}^{t}\beta_sds} du + \int_{0}^{t} \sigma T_u e^{\int_{u}^{t}\beta_sds} dW_u.
	\end{equation}
\end{lemma}
\begin{proof}
	By substituting Equation (\ref{girsanov_constant}) into Equation (\ref{proof_temperature_model_1}) and following the steps of the proof of Lemma \ref{lemma_explicit_proposed_temperature_model}, the lemma can be derived.
\end{proof}

\subsubsection{CAT Futures and Options on Futures}
Suppose that, for a contract period $[t_1, t_2]$, the temperature dynamics follow the TML model. Then, there is a price dynamic of futures written on a CAT index with $t \le t_1 < t_2$. The futures price of CAT is given by
\begin{equation}
0 = e^{-r(t_2-t)}\mathbb{E}_\mathbb{Q} \left[  \int_{t_1}^{t_2}  T_xdx - F_{CAT}(t, t_1, t_2)  \mid \mathcal{F}_t     \right].
\end{equation}

As the future price $F(t, t_1, t_2)$ is $\mathcal{F}_t$ adapted under the measure $\mathbb{Q}$,
\begin{equation}
\label{CATPRICING}
F_{CAT}(t, t_1, t_2) = \mathbb{E}_\mathbb{Q} \left[  \int_{t_1}^{t_2}  T_xdx  \mid \mathcal{F}_t     \right].
\end{equation}

\begin{proposition}
	Suppose the daily average temperature follows Model (\ref{proposed_temperature_model}). Then, the price of CAT-futures at time $t \le t_1 \le t_2$ for the contract period $[t_1, t_2]$ is given by:
	\begin{equation*}
	F_{CAT}(t, t_1, t_2)= \int_{t_1}^{t_2} S_x dx+ \int_{t_1}^{t_2} (T_t - S_t) e^{\int_{t}^{x}\beta_sds}dx + L_1,
	\end{equation*}	
	where 
	$
	L_1 = \int_{t}^{t_1} \int_{t_1}^{t_2} e^{\int_{0}^{x}\beta_sds} \sigma \lambda T_u e^{\int_{u}^{0}\beta_sds} dxdu + \int_{t_2}^{t_1} \int_{t_1}^{t_2} e^{\int_{0}^{x}\beta_sds} \sigma \lambda T_u e^{\int_{u}^{0}\beta_sds} dxdu.
	$
	\label{Proposition_CAT_Futures}
\end{proposition}

\begin{proof}
	From {Equation} (\ref{CATPRICING}), 
	\begin{equation*}
	F_{CAT}(t, t_1, t_2) = \mathbb{E}_{\mathbb{Q}}\bigg[ \int_{t_1}^{t_2} T_x dx \mid \mathcal{F}_t \bigg].
	\end{equation*}
	From Lemma \ref{lemma_explicit_girsanov_temperature_model} and for any time $x \ge t$
	\begin{equation*}
	\begin{aligned}
	T_x &= S_x + (T_t - S_t) e^{\int_{t}^{x}\beta_sds} +
	\int_{t}^{x} \sigma \lambda T_u e^{\int_{u}^{x}\beta_sds} du + \int_{t}^{x} \sigma T_u e^{\int_{u}^{x}\beta_sds} dW_u
	\\
	F_{CAT}(t, t_1, t_2) &= \mathbb{E}_{\mathbb{Q}}\bigg[ \int_{t_1}^{t_2}\bigg(S_x + (T_t - S_t) e^{\int_{t}^{x}\beta_sds} +
	\int_{t}^{x} \sigma \lambda T_u e^{\int_{u}^{x}\beta_sds} du + \int_{t}^{x} \sigma T_u e^{\int_{u}^{x}\beta_sds} dW_u \bigg) dx \bigg| \mathcal{F}_t \bigg]
	\\
	&= \int_{t_1}^{t_2} S_x dx + \int_{t_1}^{t_2} (T_t - S_t) e^{\int_{t}^{x}\beta_sds} dx +
	\int_{t_1}^{t_2} \int_{t}^{x} \sigma \lambda T_u e^{\int_{u}^{x}\beta_sds} du dx
	\\
	&= \int_{t_1}^{t_2} S_x dx + \int_{t_1}^{t_2} (T_t - S_t) e^{\int_{t}^{x}\beta_sds}dx + L_1, 
	\end{aligned}
	\end{equation*}
	where 
	\begin{equation*}
	\begin{aligned}
	L_1 &= \int_{t_1}^{t_2} \int_{t}^{x} \sigma \lambda T_u e^{\int_{u}^{x}\beta_sds} dudx 
	\\
	&= \int_{t_1}^{t_2} \int_{t}^{t_2} \mathds{1}_{[t,x]}(u) \sigma \lambda T_u e^{\int_{u}^{x}\beta_sds} dudx 
	\\
	&=  \int_{t}^{t_2} \int_{t_1}^{t_2} \mathds{1}_{[t,x]}(u) \sigma \lambda T_u e^{\int_{u}^{x}\beta_sds} dxdu 
	\\
	&= \int_{t}^{t_1} \int_{t_1}^{t_2} \mathds{1}_{[t,x]}(u) \sigma \lambda T_u e^{\int_{u}^{x}\beta_sds} dxdu + \int_{t_1}^{t_2} \int_{t_1}^{t_2} \mathds{1}_{[t,x]}(u) \sigma \lambda T_u e^{\int_{u}^{x}\beta_sds} dxdu
	\\
	&= \int_{t}^{t_1} \int_{t_1}^{t_2} e^{\int_{0}^{x}\beta_sds} \sigma \lambda T_u e^{\int_{u}^{0}\beta_sds} dxdu + \int_{u}^{t_2} \int_{t_1}^{t_2} e^{\int_{0}^{x}\beta_sds} \sigma \lambda T_u e^{\int_{u}^{0}\beta_sds} dxdu.
	\end{aligned}
	\end{equation*}
\end{proof}

\begin{proposition}\label{Proposition_CAT_Futures_inperiod}
	At time $t_1 \le t \le t_2$, the in-period (in the contract) valuation of the CAT futures is given by
	\begin{equation*}
	F_{CAT}(t, t_1, t_2) =  \int_{t_1}^{t} T_x dx + \int_{t}^{t_2} S_x dx + \int_{t}^{t_2} (T_t - S_t) e^{\int_{t}^{x}\beta_sds}dx + \int_{u}^{t_2} \int_{t}^{t_2} e^{\int_{0}^{x}\beta_sds} \sigma \lambda T_u e^{\int_{u}^{0}\beta_sds} dxdu.
	\end{equation*}
\end{proposition}
\begin{proof}
	From the CAT futures price in {Equation} (\ref{CATPRICING}), \begin{equation*}
	\begin{aligned}
	F_{CAT}(t, t_1, t_2) &= \mathbb{E}_{\mathbb{Q}}\bigg[ \int_{t_1}^{t_2} T_x dx \mid \mathcal{F}_t \bigg]
	\\
	&= \mathbb{E}_{\mathbb{Q}}\bigg[ \bigg(\int_{t_1}^{t} T_x dx + \int_{t}^{t_2} T_x dx \bigg)\bigg| \mathcal{F}_t \bigg]
	\\
	&= \int_{t_1}^{t} T_x dx + \mathbb{E}_{\mathbb{Q}}\bigg[  \int_{t}^{t_2} T_x dx \mid \mathcal{F}_t \bigg]
	\\
	&= \int_{t_1}^{t} T_x dx + F_{CAT}(t, t, t_2).
	\end{aligned}
	\end{equation*}
	
	From Proposition \ref{Proposition_CAT_Futures}, 
	\begin{equation*}
	F_{CAT}(t, t_1, t_2) =  \int_{t_1}^{t} T_x dx + \int_{t}^{t_2} S_x dx + \int_{t}^{t_2} (T_t - S_t) e^{\int_{t}^{x}\beta_sds}dx + \int_{u}^{t_2} \int_{t}^{t_2}  e^{\int_{0}^{x}\beta_sds} \sigma \lambda T_u e^{\int_{u}^{0}\beta_sds} dxdu.
	\end{equation*}
\end{proof}

\begin{lemma}
	The dynamics of the CAT futures under the equivalent probability measure $\mathbb{Q}$ and measured over the contract period $[t_1, t_2]$ is given by
	\begin{equation*}
	dF_{CAT}(t,t_1,t_2)	= \Sigma_{CAT}(t,t_1,t_2, T_t) dW_t,
	\end{equation*}
	where 
	\begin{equation*}
	\Sigma_{CAT}(t,t_1,t_2, T_t) = \sigma T_t \int_{t_1}^{t_2}  e^{\int_{t}^{x}\beta_sds} dx.
	\end{equation*}
	\label{lemma_CAT_options}
\end{lemma}
\begin{proof}
	
	We have that
	\begin{equation*}
	\begin{aligned}
	\dfrac{dF_{CAT}(t,t_1,t_2)}{dT_t} &= \int_{t_1}^{t_2} e^{\int_{t}^{x}\beta_sds} dx 
	\\	
	dF_{CAT}(t,t_1,t_2) &= \int_{t_1}^{t_2} e^{\int_{t}^{x}\beta_sds} dx  dT_t
	\\
	dF_{CAT}(t,t_1,t_2) &=  \sigma T_t \int_{t_1}^{t_2}  e^{\int_{t}^{x}\beta_sds} dx dW_t.
	\end{aligned}
	\end{equation*}
\end{proof}

\begin{proposition}\label{proposition_CAT_options}
	The call option price at exercise time $t_n$ and strike price $C$ is given by 
	\begin{equation*}
	C_{CAT}(t,t_n,t_1,t_2) = e^{-r(t_n - t)} \bigg(\big(F_{CAT}(t,t_1,t_2) - C\big) \Phi (\Delta (t,t_n,t_1,t_2, T_t)) + \Sigma_{t,t_n} \phi(\Delta (t,t_n,t_1,t_2, T_t)) \bigg),
	\end{equation*}
	where 
	\begin{equation*}
	\Delta (t,t_n,t_1,t_2, T_t) = \dfrac{F_{CAT}(t,t_1,t_2) - C}{\sqrt{\Sigma_{t,t_n}^2}}; \qquad \Sigma_{t,t_n}^2 = \int_{t}^{t_n} \Sigma_{CAT}^2(s,t_1,t_2, T_t) ds,
	\end{equation*} $\Phi$ is the cumulative standard normal distribution function, $\phi$ is the standard normal density function, and $\phi(\cdot) = \Phi'(\cdot)$.
\end{proposition}
\begin{proof}
	By definition, a call option price at exercise time $t_n$ and strike price $C$ is given as
	\begin{equation*}
	C_{CAT}(t,t_n,t_1,t_2) = e^{-r(t_n - t)} \mathbb{E}_\mathbb{Q}\big[ \max \big(F_{CAT}(t_n,t_1,t_2) - C, 0\big) \big| \mathcal{F}_t \big].
	\end{equation*}
	From Lemma \ref{lemma_CAT_options}, 
	\begin{equation*}
	\begin{aligned}
	\int_{t}^{t_n} dF_{CAT}(s,t_1,t_2)	&= \int_{t}^{t_n} \Sigma_{CAT}(s,t_1,t_2, T_t) dW_s,
	\\
	F_{CAT}(t_n,t_1,t_2) &= F_{CAT}(t,t_1,t_2) + \int_{t}^{t_n} \Sigma_{CAT}(s,t_1,t_2, T_t) dW_s.
	\end{aligned}
	\end{equation*}
	$F_{CAT}(t_n,t_1,t_2)$ is normally distributed under $\mathbb{Q}^{\lambda}$, with expectation 
	\begin{equation*}
	\mathbb{E}_{\lambda}[F_{CAT}(t_n,t_1,t_2)] = F_{CAT}(t,t_1,t_2)
	\end{equation*}
	and variance 
	\begin{equation*}
	Var_{\lambda}[F_{CAT}(t_n,t_1,t_2)] = \int_{t}^{t_n} \Sigma_{CAT}^2(s,t_1,t_2, T_t) ds = \Sigma_{t,t_n}.
	\end{equation*}
	Hence, 
	\begin{equation*}
	\begin{aligned}
	C_{CAT}(t,t_n,t_1,t_2) =& e^{-r(t_n - t)} \mathbb{E}_\mathbb{Q}\big[ \max \big(F_{CAT}(t_n,t_1,t_2) - C, 0\big) \big| \mathcal{F}_t \big]
	\\
	=& e^{-r(t_n - t)} \int_{C}^{\infty} (y - C) f_{CAT} (y) dy
	\\
	=& e^{-r(t_n - t)} \bigg(\big(F_{CAT}(t,t_1,t_2) - C\big) \Phi (\Delta (t,t_n,t_1,t_2, T_t)) + \Sigma_{t,t_n} \phi(\Delta (t,t_n,t_1,t_2, T_t)) \bigg),
	\end{aligned}
	\end{equation*}
	where $\Delta (t,t_n,t_1,t_2, T_t) = \dfrac{F_{CAT}(t,t_1,t_2) - C}{\sqrt{\Sigma_{t,t_n}^2}}$.
\end{proof}

\subsubsection{GDD Futures and Options on Futures}
Similar to the definition of the CAT future price, the GDD future price is given as
\begin{equation*}
0 = e^{-r(t_2-t)}\mathbb{E}_\mathbb{Q} \left[  \int_{t_1}^{t_2} T_xdx - F_{GDD}(t, t_1, t_2)  \biggr\rvert  \mathcal{F}_t     \right].
\label{future_price_GDD}
\end{equation*}
Using the same idea in deriving the CAT futures price, the price of the GDD-futures can be derived as
\begin{equation}
F_{GDD}(t, t_1, t_2) = \mathbb{E}_\mathbb{Q} \left[ \int_{t_1}^{t_2} \max \left(T_x - T^{optimal}, 0 \right)dx \biggr\rvert  \mathcal{F}_t  \right],
\label{GDDPRICING}
\end{equation}
\noindent
where $T^{optimal}$ is the optimal normal temperature at which a crop will develop. 

\begin{proposition}\label{proposition_GDD_futures}
	Suppose the daily average temperature follows Model (\ref{proposed_temperature_model}). Then, the price of GDD-futures at time $t \le t_1 \le t_2$ for the contract period $[t_1, t_2]$ is given by:
	\begin{equation*}
	We have that
	F_{GDD}(t,t_1,t_2) = \int_{t_1}^{t_2} \Psi(t,x) \bigg[ \phi\big({\Delta(t,x)}\big) +  \Delta(t,x)\Phi\big({\Delta(t,x)}\big),
	\bigg]
	\end{equation*}
	where 
	\begin{equation*}
	\Psi^2(t,x)= \int_{t}^{x} \sigma^2 T_u^2 e^{2 \int_{u}^{x}\beta_sds}du; \qquad \Delta(t,x) = \dfrac{S_x + (T_t -S_t) e^{\int_{t}^{x}\beta_sds} +
		\int_{t}^{x} \sigma \lambda T_u e^{\int_{u}^{x}\beta_sds} du - C}{\Psi(t,x)}.
	\end{equation*}
\end{proposition}
\begin{proof}
	By definition.
	\begin{equation}
	\begin{aligned}
	GDD &= \int_{t_1}^{t_2} \max (T_x -C, 0)dx,
	\\
	F_{GDD}(t, t_1, t_2) &= \mathbb{E}_\mathbb{Q}\bigg[ \int_{t_1}^{t_2}\max \big(T_x - C, 0\big) dx \bigg| \mathcal{F}_t \bigg].
	\end{aligned}
	\label{proof_GDD_*}
	\end{equation}
	Recall, from Lemma \ref{lemma_explicit_girsanov_temperature_model} and for any time $x \ge t$,
	\begin{equation}
	T_x = S_x + (T_t - S_t) e^{\int_{t}^{x}\beta_sds} +
	\int_{t}^{x} \sigma \lambda T_u e^{\int_{u}^{x}\beta_sds} du + \int_{t}^{x} \sigma T_u e^{\int_{u}^{x}\beta_sds} dW_u,
	\label{proof_GDD_1}
	\end{equation}
	\begin{equation}
	T_x = D_x = A(t,x) + B(t,x),
	\label{proof_GDD_2}
	\end{equation}
	where 
	\begin{equation*}
	D(t,x) = S_x + (T_t - S_t) e^{\int_{t}^{x}\beta_sds} +
	\int_{t}^{x} \sigma \lambda T_u e^{\int_{u}^{x}\beta_sds} du, \,\,\,\mbox{and}\,\,\, B(t,x) = \int_{t}^{x} \sigma T_u e^{\int_{u}^{x}\beta_sds} dW_u.
	\end{equation*}
	The distribution of $D_x$ can be determined:
	\\
	$A(t,x)$ is deterministic and, hence,
	\begin{equation*}
	B(t,x) \sim N \bigg(0, \int_{t}^{x} \sigma^2 T_u^2 e^{2 \int_{u}^{x}\beta_sds} du \bigg) = N \bigg(0, \Psi^2(t,x) \bigg).
	\end{equation*}
	It follows, from Equation (\ref{proof_GDD_2}), that 
	\begin{equation*}
	D_x \sim N\bigg( A(t,x), \Psi^2(t,x) \bigg).
	\end{equation*}
	Thus, $D_x$ can be written, in terms of the standard normal variable $Z \sim N(0,1)$, as  
	\begin{equation}
	D(t,x) = A(t,x) + \big(\Psi^2(t,x)\big)^{\frac{1}{2}}Z.
	\label{proof_GDD_3}
	\end{equation}
	Consider 
	\begin{equation*}
	T_x - C > 0.
	\end{equation*}
	This requires 
	\begin{equation*}
	\big(\Psi^2(t,x)\big)^{\frac{1}{2}}Z > C - A(t,x),
	\end{equation*}
	\begin{equation}
	Z > \dfrac{C - A(t,x)}{\big(\Psi^2(t,x)\big)^{\frac{1}{2}}Z} =  \Delta_1(t,x) .
	\label{proof_GDD_4}
	\end{equation}
	From Equation (\ref{proof_GDD_4}), 
	\begin{equation}
	C = A(t,x)  + \Delta(t,x) \big(\Psi^2(t,x)\big)^{\frac{1}{2}}Z.
	\label{proof_GDD_5}
	\end{equation}
	From Equations (\ref{proof_GDD_*}) and (\ref{proof_GDD_4}),
	\begin{equation}
	\mathbb{E}_\mathbb{Q}\bigg[ \int_{t_1}^{t_2}\max \big(T_x - C, 0\big) dx \bigg| \mathcal{F}_t \bigg] = \int_{\Delta(t,x)}^{+\infty} \bigg( D(t,x) - C \bigg) \dfrac{e^{-\frac{1}{2}z^2}}{\sqrt{2\pi}}dz.
	\label{proof_GDD_6}
	\end{equation}
	Substituting Equations (\ref{proof_GDD_3}) and (\ref{proof_GDD_5}) into Equation (\ref{proof_GDD_6}), 
	\begin{equation*}
	\begin{aligned}
	&=  \int_{\Delta(t,x)}^{+\infty} \bigg( A(t,x) + \big(\Psi^2(t,x)\big)^{\frac{1}{2}}Z - A(t,x)  - \Delta(t,x) \big(\Psi^2(t,x)\big)^{\frac{1}{2}}Z \bigg) \dfrac{e^{-\frac{1}{2}z^2}}{\sqrt{2\pi}}dz
	\\
	&=\int_{\Delta(t,x)}^{+\infty} \bigg( \big(\Psi^2(t,x)\big)^{\frac{1}{2}}z - \Delta(t,x) \big(\Psi^2(t,x)\big)^{\frac{1}{2}}z \bigg) \dfrac{e^{-\frac{1}{2}z^2}}{\sqrt{2\pi}}dz
	\\
	&=  \left(\Psi^2(t,x)\right)^{\frac{1}{2}} \left( \int_{\Delta_1(t,x)}^{+\infty} \dfrac{ze^{-\frac{1}{2}z^2}}{\sqrt{2\pi}}dz + \Delta_1(t,x) \Phi(-\Delta_1(t,x)) \right )
	\\
	&= \left(\Psi^2(t,x)\right)^{\frac{1}{2}} \left(\dfrac{e^{-\frac{1}{2}{\Delta(t,x)}^2}}{\sqrt{2\pi}}dz + \Delta(t,x) \Phi(\Delta(t,x)) \right )
	\\
	&= 	\left(\Psi^2(t,x)\right)^{\frac{1}{2}}
	\big[\phi(\Delta(t,x)) + \Delta(t,x) \Phi(\Delta(t,x)) \big],
	,		
	\end{aligned}
	\end{equation*}
	where
	\begin{equation*}
	\Delta(t,x) = -\Delta_1(t,x) = \dfrac{G(t,x) - C}{\left(\Psi^2(t,x)\right)^{\frac{1}{2}}},
	\end{equation*}
	
	\begin{equation}
	\mathbb{E}_\mathbb{Q}\bigg[ \int_{t_1}^{t_2}\max \big(T_x - C, 0\big) dx \bigg| \mathcal{F}_t \bigg] = \Psi(t,x) \bigg[ \phi\big(\Delta(t,x)\big) + \Delta(t,x)\Phi\big(\Delta(t,x)\big)\bigg]
	\label{proof_GDD_7}.
	\end{equation}
	Substituting Equation (\ref{proof_GDD_7}) into Equation (\ref{proof_GDD_*}) gives the Proposition. 
\end{proof}

\begin{lemma}
	The dynamics of the GDD futures under the equivalent probability measure $\mathbb{Q}$ measured over the period $[t_1, t_2]$ are given by
	\begin{equation*}
	dF_{GDD}(t,t_1,t_2) = \varPi_{GDD}(t,t_1,t_2)dW_t, 
	\end{equation*}
	\noindent where $\varPi_{GDD}$ is called the term structure of the GDD futures volatility, 
	\begin{equation*}
	\varPi_{GDD} = \sigma T_t \int_{t_1}^{t_2} e^{\int_{t}^{x}\beta_sds}  \Phi \bigg(\dfrac{h(t,x,e^{\int_{t}^{x}\beta_sds}(T_t -S_t))}{\Psi(t,x)}\bigg) ds; \quad  \Psi^2(t,x)= \int_{t}^{x} \sigma^2 T_u^2 e^{2 \int_{u}^{x}\beta_sds}du;
	\end{equation*}
	\begin{equation*}
	h(t,x,e^{\int_{t}^{x}\beta_sds}(T_t -S_t)) = S_x + (T_t -S_t) e^{\int_{t}^{x}\beta_sds} +
	\int_{t}^{x} \sigma \lambda T_u e^{\int_{u}^{x}\beta_sds} du - C.
	\end{equation*}
\end{lemma}
\begin{proof}
	Let $h(t,x,e^{\int_{t}^{x}\beta_sds}(T_t -S_t)) = S_x + (T_t -S_t) e^{\int_{t}^{x}\beta_sds} +
	\int_{t}^{x} \sigma \lambda T_u e^{\int_{u}^{x}\beta_sds} du - C$
	\begin{equation*}
	\dfrac{dh}{dT_t} =   e^{\int_{t}^{x}\beta_sds}
	\end{equation*}
	From Proposition \ref{proposition_GDD_futures},
	\begin{equation*}
	\begin{aligned}
	\dfrac{dF_{GDD}}{dT_t} &= \int_{t_1}^{t_2} \Psi(t,x) \Upsilon'\bigg(\dfrac{h(t,x,e^{\int_{t}^{x}\beta_sds}(T_t -S_t))}{\Psi(t,x)}\bigg)\dfrac{h'(t,x,e^{\int_{t}^{x}\beta_sds}(T_t -S_t))}{\Psi^2(t,x)}ds
	\\
	&= \int_{t_1}^{t_2} \Upsilon' \bigg(\dfrac{h(t,x,e^{\int_{t}^{x}\beta_sds}(T_t -S_t))}{\Psi(t,x)}\bigg) h'(t,x,e^{\int_{t}^{x}\beta_sds}(T_t -S_t)) ds
	\\
	&= \int_{t_1}^{t_2} e^{\int_{t}^{x}\beta_sds}  \Phi \bigg(\dfrac{h(t,x,e^{\int_{t}^{x}\beta_sds}(T_t -S_t))}{\Psi(t,x)}\bigg) ds
	\\
	dF_{GDD}&= \sigma T_t \int_{t_1}^{t_2} e^{\int_{t}^{x}\beta_sds}  \Phi \bigg(\dfrac{h(t,x,e^{\int_{t}^{x}\beta_sds}(T_t -S_t))}{\Psi(t,x)}\bigg) dsdW_t
	\\
	dF_{GDD} &= \varPi_{GDD}(t,t_1,t_2)dW_t, 
	\end{aligned}
	\end{equation*}
	where 
	\begin{equation*}
	\begin{aligned}
	\Upsilon(\Delta(t,x)) &= \phi(\Delta(t,x)) + \Delta(t,x) \Phi(\Delta(t,x)); \quad \Delta(t,x) = \dfrac{h(t,x,e^{\int_{t}^{x}\beta_sds}(T_t -S_t))}{\Psi(t,x)}; 
	\\
	\varPi_{GDD} &= \sigma T_t \int_{t_1}^{t_2} e^{\int_{t}^{x}\beta_sds}  \Phi \bigg(\dfrac{h(t,x,e^{\int_{t}^{x}\beta_sds}(T_t -S_t))}{\Psi(t,x)}\bigg) ds.
	\end{aligned}
	\end{equation*}
\end{proof}

\begin{proposition}
	For a strike price $C$ and maturity time $t \le t_n \le t_1$, the price of a call option at time $t$ on a GDD futures contract is given by
	\begin{equation*}
	C_{GDD}(t,t_n,t_1,t_2) = e^{-r(t_n - t)} \mathbb{E}_{\mathbb{Q}}\bigg[\max\bigg( \int_{t_1}^{t_2} \Psi(t,x) P(t,x,t_n,(T_t - S_t))ds - C, 0 \bigg)\bigg],
	\end{equation*} 
	where 
	\begin{equation*}
	P(t,x,t_n,(T_t - S_t)) = \tilde{\Upsilon}\bigg( t,x, e^{\int_{t}^{x}\beta_sds}(T_t - S_t) + \int_{t}^{t_n} \lambda \sigma e^{\int_{u}^{x}\beta_sds}du + \Sigma(x, t, t_n)Y \bigg)
	\end{equation*}
	\begin{equation*}
	\tilde{\Upsilon} (t,x,e^{\int_{t}^{x}\beta_sds}(T_t -S_t)) =  \Upsilon \bigg( \dfrac{h(t,x,e^{\int_{t}^{x}\beta_sds}(T_t -S_t))}{\Psi(t,x)} \bigg)
	\end{equation*}
	and 
	\begin{equation*}
	\Sigma(x, t, t_n) = \int_{t}^{t_n} \sigma^2 T_u^2 e^{\int_{u}^{x}\beta_sds} du.
	\end{equation*}
\end{proposition}
\begin{proof}
	By definition ,
	\begin{equation*}
	\begin{aligned}
	C_{GDD}(t,t_n,t_1, t_2) &= e^{-r(t_n -t )}\mathbb{E}_\mathbb{Q}\bigg[ \int_{t_1}^{t_2}\max \big(F_{GDD}(t_n,t_1,t_2) - C, 0\big) dx \bigg| \mathcal{F}_t \bigg]
	\\
	F_{GDD}(t_n,t_1, t_2) &= \int_{t_1}^{t_2} \Psi(t,x){\tilde{\Upsilon}}\big(t,x,e^{\int_{t_n}^{x}\beta_sds}(T_{t_n} -S_{t_n})\big)ds,
	\\
	&= \int_{t_1}^{t_2} \Psi(t,x)  {\tilde{\Upsilon}}\bigg(t,x,e^{\int_{t_n}^{x}\beta_sds}(T_t -S_t) + \int_{t}^{t_n} \lambda \sigma T_u e^{\int_{u}^{x}\beta_sds}du + \int_{t}^{t_n} \sigma T_u  e^{\int_{u}^{x}\beta_sds} dW_u\bigg)ds
	\\
	C_{GDD}(t,t_n,t_1, t_2) &= e^{-r(t_n -t )} \mathbb{E}_\mathbb{Q}\bigg[ \int_{t_1}^{t_2}\max \big(\int_{t_1}^{t_2} \Psi(t,x)  {\tilde{\Upsilon}}\bigg(t,x,e^{\int_{t_n}^{x}\beta_sds}(T_t -S_t) + \int_{t}^{t_n} \lambda \sigma T_u e^{\int_{u}^{x}\beta_sds}du + \\
	&\int_{t}^{t_n} \sigma T_u  e^{\int_{u}^{x}\beta_sds} dW_u\bigg)ds - C, 0\big) dx \bigg| \mathcal{F}_t \bigg],
	\end{aligned}
	\end{equation*}
	where
	\begin{equation*}
	\tilde{\Upsilon} (t,x,e^{\int_{t}^{x}\beta_sds}(T_t -S_t)) =  \Upsilon \bigg( \dfrac{h(t,x,e^{\int_{t}^{x}\beta_sds}(T_t -S_t))}{\Psi(t,x)} \bigg).
	\end{equation*}
\end{proof}

\subsubsection{CAT and GDD Futures on Temperature Basket}
\noindent
Assuming $N$ is the spatial locations in the basket, then $\left(\omega_{i}\right)_{i=1}^N$ will be the collection of weights for the spatial locations $\left(y_{i}\right)_{i=1}^N$. The basket of the daily average temperature at the $N$ spatial locations for a given time $t$ is defined as:
\begin{equation}
M(t) := \sum_{i=1}^{N} \omega_{i} {T}_i(t), 
\label{basket_spatial_location}
\end{equation}
where $\sum_{i=1}^{N} \omega_{i} =1$.
\\
[2mm]
Assume the daily average temperature is spatially correlated across the random noise term and the risk-neutral distribution of the daily average temperature for the spatial locations is normally distributed in Model (\ref{proposed_temperature_model}). Hence, the weighted sum of a normally distributed basket is also normally distributed. 
From the above assumptions, a new daily average temperature model for each spatial location $y^i$ can be proposed,  
\begin{equation}
dT_i(t) = dS_i(t) + \beta_i(t) \big(T_i(t) - S_i(t)\big)dt + \sigma T_i(t)dB_i(t).
\label{basket_deseasonalizedmodel}
\end{equation}
Expressing Equation (\ref{basket_deseasonalizedmodel}), for locations $i=1,2,3\cdots,N$, as an $N$-dimensional system leads to the equation below, 
\begin{equation}
d\bm{T}(t) = d\bm{S}(t) + \bm{\beta}(t) \big(\bm{T}(t) - \bm{S}(t)\big)dt + \bm{\sigma} \bm{T}(t)d\bm{B}(t),
\label{basket_Ndimensional_model}
\end{equation}
where $\bm{B}(t) \sim N(0, \bm{\Omega} t)$ and $\bm{\Omega}$ is a covariance matrix. 
Using the linear transformation of multivariate normal distributions property, 
\begin{equation*}
Y \sim N(\mu, \Sigma) \Rightarrow DY \sim N(D\mu, D\Sigma D^T).
\end{equation*}
Suppose $Z \sim N(0, \bm{I}t)$ and $Y = DZ$. Then, it follows that $Y \sim N(0, DD^Tt)$. Applying Cholesky factorization to $\Sigma$, we can derive a lower triangular form for $D$, and $\bm{W}_t$ will be expressed as an $N$-dimensional Brownian motion $\bm{V}_t$, 
\begin{equation}
\bm{B}(t) = \bm{LV}(t).
\label{n_dimensional_Brownian_motion}
\end{equation}
Then, $\bm{LL^T} = \bm{\Omega}$, $\bm{L}$ is a lower triangular matrix with non-negative diagonal entries, $\bm{L^T}$ is an upper triangular matrix, and $\bm{V}(t) = (V_1(t), V_2(t), V_3(t), \cdots, V_N(t))^T$ with $dV_i(t) dV_j(t) = \delta_{ij}dt$. 
Equation~(\ref{basket_Ndimensional_model}) can be reformulated, in terms of $\bm{V}(t)$, as 
\begin{equation}	\label{basket_N_dimensional_model_reformulated}
d\bm{T}(t) = d\bm{S}(t) + \bm{\beta}(t) \big(\bm{T}(t) - \bm{S}(t)\big)dt + \bm{\sigma} \bm{T}(t)\bm{L}d\bm{V}(t).
\end{equation}

\subsubsection{Girsanov's Theorem in $\mathbb{R}^N$}
Let $\bm{V}(t)=\big(V_1(t), V_2(t), V_3(t), \cdots, V_N(t)\big)$ be an $N$-dimensional Brownian motion on a probability space $(\Omega, \mathcal{F}, \mathbb{P})$ and $\bm{\lambda} = \big(\lambda_1(t), \lambda_2(t), \lambda_3(t), \cdots, \lambda_N(t)\big)$ be an $N$-dimensional adapted process on $[0, T]$. 

Define 
\begin{equation}
Z_{\bm{\lambda}}(t) := \exp \bigg( \int_{0}^{t} \bm{\lambda}(s)d\bm{V}(s) - \dfrac{1}{2} \int_{0}^{t} \mid\mid\bm{\lambda}(s)\mid\mid^2 ds \bigg),
\label{spatial_radon_nikodym_1}
\end{equation}
where $\mid\mid\bm{\lambda}(s)\mid\mid^2 = \sum_{i=1}^{N} \lambda_i(s)^2$.

Let 
\begin{equation}
\bm{\tilde{V}}(t) = \bm{V}(t) + \int_{0}^{t}\bm{\lambda}(s)ds.
\label{girsanov_1}
\end{equation}

The component process of $\bm{\tilde{B}}(t)$ is independent under the measure $\mathbb{Q}$.  

Suppose that 
\begin{equation*}
\mathbb{E}\int_{0}^{t}\mid\mid\lambda(s)\mid\mid^2Z(s)^2ds < \infty.
\end{equation*}
Then, $\bm{\tilde{B}}(t)$ is an $N$-dimensional standard Brownian motion under the measure $\mathbb{Q}$, defined as 
\begin{equation}
\dfrac{d\mathbb{Q}}{d\mathbb{P}} \bigg| \mathcal{F}_t = Z(T).
\label{spadial_radon_nikodyn_2}
\end{equation}

Let
\begin{equation}
Z_{\bm{\lambda}}(t) := \exp \bigg(\int_{0}^{t} \big( \bm{\sigma} \bm{T}(t) L \big) \bm{\lambda}(s) d\bm{V}(s) - \dfrac{1}{2} \int_{0}^{t} \mid\mid \bm{\sigma} \bm{T}(s)\bm{L}\mid\mid^{-2} \mid\mid \bm{\lambda}(s)\mid\mid^2 ds\bigg).
\label{spatial_radon_nikodym_3}
\end{equation}

For a constant market price of risk at each geographical reference location in equation (\ref{spatial_radon_nikodym_3}),  it can be deduced that
\begin{equation}
\bm{W}_{\lambda}(t) = \bm{V}(t) - \int_{0}^{t} (\bm{\sigma} \bm{T}(s) \bm{L})^{-1}\bm{\lambda} ds,
\label{spatial_radon_nikodym_4}
\end{equation}
where $\bm{W}_{\lambda}(t)$ is a standard Brownian motion under the measure $\mathbb{Q}$. Hence, we can define the temperature model under the measure $\mathbb{Q}$ as 
\begin{equation}
d\bm{T}(t) = d\bm{S}(t) + \big[\bm{\lambda} + \bm{\beta}(t) \big(\bm{T}(t) - \bm{S}(t)\big)\big]dt + \bm{\sigma} \bm{T}(t)\bm{L}d\bm{W}(t).
\label{spatial_new_temperature_model}
\end{equation}

\subsubsection{Pricing CAT and GDD Futures on Temperature Basket}
\begin{definition}
	For a a specificied contract period, $t\le t_1 < t_2$ and at a spatial location $y_i$, the CAT futures price is defined as in Equation (\ref{Futures_Basket_spatial_CAT}):
	\begin{equation}
	\begin{aligned}
	CAT_M[t_1, t_2] &:= \int_{t_1}^{t_2} M(t)
	\\
	&:= \sum_{i=1}^{N} \omega_{i} \left(\int_{t_1}^{t_2} T_i(x) dx\right).
	\end{aligned}
	\label{Basket_spatial_CAT_definition}
	\end{equation}
\end{definition}
From Equations (\ref{CATPRICING}), (\ref{Basket_spatial_CAT_definition}), and the linearity of expectation, 
\begin{equation}
F_{CAT}(t,t_1, t_2; M) = \sum_{i=1}^{N} \omega_{i} \mathbb{E}_{\mathbb{Q}} \left[\int_{t_1}^{t_2} T_i(x) dx  \mid \mathcal{F}_t \right].
\label{Futures_Basket_spatial_CAT}
\end{equation}

\begin{definition}\label{definition_Basket_spatial_GDD_definition}
	For a a specificied contract period, $t\le t_1 < t_2$ and at location $y_i$, the CAT futures price is defined as
	\begin{equation}
	\begin{aligned}
	\label{GDD_defn}
	GDD(\tau_1, \tau_2) &:= \int_{t_1}^{t_2} max \Big\{  M(t) - C\,,\,0 \Big\} dx
	\\
	&= 	\int_{t_1}^{t_2} max  \bigg\{   \sum_{i=1}^{N} \omega_{i} T_i(x)  - C\,,\,0 \bigg\} dx.
	\end{aligned}
	\end{equation}
\end{definition}

From Equations (\ref{definition_Basket_spatial_GDD_definition}), (\ref{GDDPRICING}), and using the linearity of expectation,
\begin{equation}
\begin{aligned}
F_{GDD}(t,t_1,t_2; D) &= \mathbb{E}_{\mathbb{Q}} \Bigg( \int_{t_1}^{t_2} max  \bigg\{   \sum_{i=1}^{N} \omega_{i} T_i(x)  - C\,,\,0 \bigg\} dx \biggr\rvert  \mathcal{F}_t \Bigg)
\\
&= \int_{t_1}^{t_2} \mathbb{E}_{\mathbb{Q}} \Bigg(  max  \bigg\{   \sum_{i=1}^{N} \omega_{i} T_i(x)  - C \,,\,0 \bigg\} \biggr\rvert  \mathcal{F}_t \Bigg)dx.
\end{aligned}
\label{Spatial_GDD_pricing}
\end{equation}

\begin{lemma}
	If the dynamics of the daily average temperature follows Equation (\ref{spatial_new_temperature_model}), then the explicit solution for the $i^{th}$ location $y_i$ is given as 
	\begin{equation*}
	T_i(t) = S_i(t) + \big(T_i(0) -S_i(0)\big) e^{\int_{0}^{t}\beta_i(s)ds} +
	\int_{0}^{t} \lambda_i e^{\int_{u}^{t}\beta_i(s)ds} du + \int_{0}^{t} \sigma_i T_i(u) e^{\int_{u}^{t}\beta_i(s)ds} \sum_{j=1}^{i}L_{ij}dW_j^{\lambda}(u).
	\end{equation*}
	\label{lemma_spatial_CAT_pricing}
\end{lemma}
\begin{proof}
	The proof follows directly from the proof of Lemma \ref{lemma_explicit_proposed_temperature_model} and observing the $i^{th}$ location $y_i$. 
\end{proof}

\begin{proposition}\label{proposition_mean-reverting_CAT_basket_futures}
	At a spatial location $i$, the futures contract price on basket of CAT index following the mean-reverting regime in Equation (\ref{basket_deseasonalizedmodel}) is calculated as
	\begin{equation*}
	\begin{aligned}
	F_{CAT}(t,t_1, t_2; M) = \sum_{i=1}^{N} \omega_{i}  &\bigg[ \int_{t_1}^{t_2} S_i(x) + \int_{t_1}^{t_2} \big(T_i(t) -S_i(t)\big) e^{\int_{t}^{x}\beta_i(s)ds} + 
	\\
	&\int_{t}^{t_1} \int_{t_1}^{t_2} e^{\int_{0}^{x}\beta_i(s)ds} \lambda_i e^{\int_{u}^{0}\beta_i(s)ds} dxdu + 
	\\
	&\int_{u}^{t_2} \int_{t_1}^{t_2} e^{\int_{0}^{x}\beta_i(s)ds} \lambda_i e^{\int_{u}^{0}\beta_i(s)ds} dxdu \bigg].
	\end{aligned}
	\end{equation*}
\end{proposition}
\begin{proof}
	For $x \ge t$ in Lemma \ref{lemma_spatial_CAT_pricing}, 
	\begin{equation}
	T_i(x) = S_i(x) + \big(T_i(t) -S_i(t)\big) e^{\int_{t}^{x}\beta_i(s)ds} +
	\int_{t}^{x} \lambda_i e^{\int_{u}^{x}\beta_i(s)ds} du + \int_{t}^{x} \sigma_i T_i(u) e^{\int_{u}^{x}\beta_i(s)ds} \sum_{j=1}^{i}L_{ij}dW_j^{\lambda}(u)
	\label{Spatial_CAT_x_ge_t}
	\end{equation}
	
	\begin{equation*}
	\begin{aligned}
	\mathbb{E}_{\mathbb{Q}} \big[\int_{t_1}^{t_2} T_i(x) dx  \mid \mathcal{F}_t \big] &= \mathbb{E}_{\mathbb{Q}}\bigg[\int_{t_1}^{t_2} \bigg( 
	S_i(x) + \big(T_i(t) -S_i(t)\big) e^{\int_{t}^{x}\beta_i(s)ds} +
	\int_{t}^{x} \lambda_i e^{\int_{u}^{x}\beta_i(s)ds} du + \\ &\int_{t}^{x} \sigma_i T_i(u) e^{\int_{u}^{x}\beta_i(s)ds} \sum_{j=1}^{i}L_{ij}dW_j^{\lambda}(u)\bigg)dx \big| \mathcal{F}_t \bigg]
	\\
	&= \int_{t_1}^{t_2} S_i(x) dx + \int_{t_1}^{t_2} \big(T_i(t) -S_i(t)\big) e^{\int_{t}^{x}\beta_i(s)ds} dx +
	\int_{t_1}^{t_2}\int_{t}^{x} \lambda_i e^{\int_{u}^{x}\beta_i(s)ds} dudx 
	\\
	&= \int_{t_1}^{t_2} S_i(x) + \int_{t_1}^{t_2} \big(T_i(t) -S_i(t)\big) e^{\int_{t}^{x}\beta_i(s)ds} + L_1, 
	\end{aligned}
	\end{equation*}
	where 
	\begin{equation*}
	\begin{aligned}
	L_1 &= \int_{t_1}^{t_2}\int_{t}^{x} \lambda_i e^{\int_{u}^{x}\beta_i(s)ds} dudx
	\\
	&= \int_{t_1}^{t_2} \int_{t}^{t_2} \mathds{1}_{[t,x]}(u) \lambda_i e^{\int_{u}^{x}\beta_i(s)ds} dudx 
	\\
	&= \int_{t}^{t_1} \int_{t_1}^{t_2} e^{\int_{0}^{x}\beta_i(s)ds} \lambda_i e^{\int_{u}^{0}\beta_i(s)ds} dxdu + \int_{u}^{t_2} \int_{t_1}^{t_2} e^{\int_{0}^{x}\beta_i(s)ds} \lambda_i e^{\int_{u}^{0}\beta_i(s)ds} dxdu
	\end{aligned}
	\end{equation*}
	\begin{equation*}
	\begin{aligned}
	F_{CAT}(t,t_1, t_2; M) = \sum_{i=1}^{N} \omega_{i}  &\bigg[ \int_{t_1}^{t_2} S_i(x) + \int_{t_1}^{t_2} \big(T_i(t) -S_i(t)\big) e^{\int_{t}^{x}\beta_i(s)ds} + 
	\\
	&\int_{t}^{t_1} \int_{t_1}^{t_2} e^{\int_{0}^{x}\beta_i(s)ds} \lambda_i e^{\int_{u}^{0}\beta_i(s)ds} dxdu + 
	\\
	&\int_{u}^{t_2} \int_{t_1}^{t_2} e^{\int_{0}^{x}\beta_i(s)ds} \lambda_i e^{\int_{u}^{0}\beta_i(s)ds} dxdu \bigg].
	\end{aligned}
	\end{equation*}
\end{proof}

\begin{proposition}\label{proposition_Spatial_GDD_normal}
	The GDD futures price for a contract time $t \le t_1 < t_2$ at a given spatial location $y^i$ for a basket of CAT index following the normal regime is given by
	\begin{equation}
	F_{GDD}^{N}(t,t_1,t_2; D)  = \int_{t_1}^{t_2} \big( \xi (t,x) + 2\varDelta(t,x) \big)^{\ \frac{1}{2}} \big(\phi \big(\Lambda(t,x)\big) + \Lambda(t,x)\Phi\big(\Lambda(t,x)\big)\big)dx,
	\end{equation}
	where $\Phi$ is the cumulative standard normal distribution function, $\phi$ is the standard normal density function,	
	\begin{equation*}
	\Lambda(t,x) = \dfrac{\psi (t,x)- C}{ \bigg(\xi(t,x) + 2 \varDelta(t,x)  \bigg)^{\frac{1}{2}}},
	\end{equation*}	
	\begin{equation*}
	\psi(t,x)  =  \sum_{i=1}^{N} \omega^{i}\bigg( S_i(x) + \big(T_i(t) -S_i(t)\big) e^{\int_{t}^{x}\beta_i(s)ds} +
	\int_{t}^{x} \lambda_i e^{\int_{u}^{x}\beta_i(s)ds} du
	\bigg),
	\end{equation*}	
	\begin{equation*}
	\xi(t,x) %= \sum_{i=1}^{N} \omega^{2}_i \Psi^2(t,x) 
	= \sum_{i=1}^{N} \omega^{2}_i \sum_{j=1}^{i}\displaystyle\int_{t}^{x} \sigma_i^2 T_i(u)^2 e^{2\int_{u}^{x}\beta_i(s)ds} L^2_{ij} du,\text{ and}
	\end{equation*}
	\begin{equation*}
	\varDelta(t,x)	= \sum_{i=1}^{N} \sum_{j=i+1}^{N} \omega^{i} \omega^{j} \Big( \sum_{q=1}^{i}L^{iq}L^{jq} \Big) \displaystyle\int_{t}^{x} \sigma_i(u) \sigma_j(u) {T}_i(u) {T}_j(u) e^{\int_{u}^{x}(\beta_i(s) + \beta_j(s))ds} du
	%	\end{aligned}
	\end{equation*} 
\end{proposition}
\begin{proof}
	Let
	\begin{equation}
	D(x) = \sum_{i=1}^{N} \omega^{i} \tilde{T}_t^i.
	\end{equation}
	For convenience, we denote the deterministic and random components of Equation (\ref{Spatial_CAT_x_ge_t}) as $A_i(t,x)$ and $B_i(t,x)$, respectively. That is,
	\begin{equation*}
	\begin{aligned}
	A_i(t,x) &= S_i(x) + \big(T_i(t) -S_i(t)\big) e^{\int_{t}^{x}\beta_i(s)ds} +
	\int_{t}^{x} \lambda_i e^{\int_{u}^{x}\beta_i(s)ds} du,
	\\
	B_i(t,x) &=\int_{t}^{x} \sigma_i T_i(u) e^{\int_{u}^{x}\beta_i(s)ds} \sum_{j=1}^{i}L_{ij}dW_j^{\lambda}(u) =\sum_{j=1}^{i} \int_{t}^{x} \sigma_i T_i(u) e^{\int_{u}^{x}\beta_i(s)ds} L_{ij}dW_j^{\lambda}(u).
	\end{aligned}
	\end{equation*}
	Hence,
	\begin{equation}
	D(x) = \sum_{i=1}^{N} \omega_{i}\Big(A_i(t,x) + B_i(t,x)\Big).
	\label{spatial_normal_split}
	\end{equation}	
	The distribution of the basket $D(x)$ at time $t$ wll be established. $A_i(t,x)$ is, however, deterministic. By It\^{o} isometry, and at $t$, 
	\begin{equation*}
	\int_{t}^{x} \sigma_i T_i(u) e^{\int_{u}^{x}\beta_i(s)ds} L_{ij}dW_j^{\lambda}(u) \sim N\bigg(0, \displaystyle\int_{t}^{x} \sigma_i^2 T_i(u)^2 e^{2\int_{u}^{x}\beta_i(s)ds} L^2_{ij} du \bigg).
	\end{equation*}
	However, the $W^{\lambda}_j(u)$ are independent for each $j$. The variances can, therefore, be summed to obtain the variance of $B_i(t,x)$, 
	\begin{equation*}
	B_i(t,x) \sim N\bigg(0, \sum_{j=1}^{i}\displaystyle\int_{t}^{x} \sigma_i^2 T_i(u)^2 e^{2\int_{u}^{x}\beta_i(s)ds} L^2_{ij} du \bigg) = N\bigg(0, \Psi^2(t,x) \bigg).
	\end{equation*}
	As $\sum_{i=1}^{N} \omega_{i} B^i(t,x)$ is a sum of normally distributed random variables, it is normally distributed with respective mean and variance:
	\begin{equation*}
	\mathbb{E}\Big( \sum_{i=1}^{N} \omega_{i} B^i(t,x) \Big) = \sum_{i=1}^{N} \omega_{i} \mathbb{E}\Big(B_i(t,x)\Big) = 0,
	\end{equation*}
	and 
	\begin{equation*}
	\begin{aligned}
	Var \Big( \sum_{i=1}^{N} \omega_{i} B_i(t,x)\Big) &= \sum_{i=1}^{N} Var\Big(\omega_{i} B_i(t,x)\Big) + 2 \sum_{i < j} Cov\Big( \omega_{i}H^i, \omega_{j}H^j \Big)
	\\
	&=\sum_{i=1}^{N} \omega^{2}_i Var(H_i(t,x))+ 2 \sum_{i <j}  \omega_{i}  \omega_{j} Cov\Big(H_i, H_j\Big)
	\\
	&=\sum_{i=1}^{N} \omega^{2}_i \Psi^2(t,x)+ 2 \sum_{i <j}  \omega_{i}  \omega_{j} Cov\Big(B_i, B_j\Big).
	\end{aligned}
	\end{equation*}
	Consider $Cov\big(B_i, B_j\big)$ for $j > 1$. With respect to the standard Brownian motions $W_1^{\lambda}(u)$, both $B_1$ and $B_j$ are in the same integral form. For $j>1$ and by the independence of $W_1^{\lambda}(u)$ and $W_j^{\lambda}(u)$, the covariance only exists between these two integrals. Therefore, 	
	\begin{equation*}
	\begin{aligned}
	Cov\Big(B_i, B_j\Big) &= \displaystyle\int_{t}^{x} \sigma_i(u) \sigma_j(u) {T}_i(u) {T}_j(u) 
	e^{\int_{u}^{x}(\beta_i(s) + \beta_j(s))ds}
	\Big( \sum_{q=1}^{i}L^{iq}L^{jq} \Big) du, \quad \forall j > 1
	\\
	&= \Big( \sum_{q=1}^{i}L^{iq}L^{jq} \Big) \displaystyle\int_{t}^{x} \sigma_i(u) \sigma_j(u) {T}_i(u) {T}_j(u) e^{\int_{u}^{x}(\beta_i(s) + \beta_j(s))ds} du, \quad \forall j > 1.
	\end{aligned}
	\end{equation*}
	Define $\varUpsilon{ij}(t,x) := \displaystyle\int_{t}^{x} \sigma_i(u) \sigma_j(u) {T}_i(u) {T}_j(u) e^{\int_{u}^{x}(\beta_i(s) + \beta_j(s))ds} du$,
	\begin{equation*}
	\begin{aligned}
	Var \Big( \sum_{i=1}^{N} \omega_{i} B_i(t,x)\Big) &= \sum_{i=1}^{N} \omega^{2}_i \Psi^2(t,x) + 2 \sum_{i <j}  \omega_{i}  \omega_{j} \Big( \sum_{q=1}^{i}L_{iq}L_{jq} \Big) \varUpsilon{ij}(t,x) 
	\\
	&= \sum_{i=1}^{N} \omega^{2}_i \Psi^2(t,x) + 2 \sum_{i=1}^{N}\sum_{j=i+1}^{N} \omega_{i} \omega_{j} \Big( \sum_{q=1}^{i}L_{iq}L_{jq} \Big) \varUpsilon_{ij}(t,x).
	\end{aligned}
	\end{equation*}
	From (\ref{spatial_normal_split}), 
	\begin{equation*}
	D(x) \sim N \bigg( \sum_{i=1}^{N} \omega_{i}A_i(t,x) ,  \sum_{i=1}^{N} \omega^{2}_i \Psi^2(t,x) + 2 \sum_{i=1}^{N}\sum_{j=i+1}^{N} \omega_{i} \omega_{j} \Big( \sum_{q=1}^{i}L_{iq}L_{jq} \Big) \varUpsilon_{ij}(t,x)   \bigg).
	\end{equation*}
	Let 
	\begin{equation*}
	\psi (t,x) = \sum_{i=1}^{N} \omega_{i}A_i(t,x); \quad \xi(t,x) = \sum_{i=1}^{N} \omega^{2}_i \Psi^2(t,x);
	\end{equation*}
	\begin{equation*}
	\varDelta(t,x) = \sum_{i=1}^{N}\sum_{j=i+1}^{N} \omega_{i} \omega_{j} \Big( \sum_{q=1}^{i}L_{iq}L_{jq} \Big) \varUpsilon_{ij}(t,x).
	\end{equation*}
	$D(x)$ can be written, in the form of a standard normal random variable $Z \sim N(0,1)$, as
	\begin{equation}
	D(x) = \psi (t,x)  + \bigg(\xi(t,x) + 2 \varDelta(t,x)  \bigg)^{\frac{1}{2}}Z.
	\label{spatial_standard_normal_GDD_normal}
	\end{equation}
	Recall Equation (\ref{Spatial_GDD_pricing}) and consider
	\begin{equation*}
	\begin{aligned}
	\sum_{i=1}^{N} \omega_{i} {T}_i(x) - C &> 0	,
	\\
	\bigg(\xi(t,x) + 2 \varDelta(t,x)  \bigg)^{\frac{1}{2}}Z &> C - \psi (t,x),
	\\
	Z > \dfrac{C - \psi (t,x)}{ \bigg(\xi(t,x) + 2 \varDelta(t,x)  \bigg)^{\frac{1}{2}}} &:= \Lambda^{1}(t,x).
	\end{aligned}
	\end{equation*}
	From the above equation,
	\begin{equation}
	C = \psi(t,x) + \Lambda^{1}(t,x){\Big(\xi(t,x) + 2 \varDelta(t,x)  \Big)^{\frac{1}{2}}}.
	\label{spatial_normal_reformulate_normal}
	\end{equation}	
	It can be deduced from, Equations (\ref{Spatial_GDD_pricing}) and  (\ref{spatial_normal_reformulate_normal}), that
	\begin{equation}
	\mathbb{E}_{\mathbb{Q}} \Bigg(  max  \bigg\{\sum_{i=1}^{N} \omega_{i} {T}_i(x)  - C \,,\,0 \bigg\} dx \biggr\rvert  \mathcal{F}_t \Bigg) = \int_{\Lambda^{1}(t,x)}^{+\infty} \Big( D(x) - C \Big) \dfrac{e^{-\frac{1}{2}z^2}}{\sqrt{2\pi}} dz.
	\end{equation}
	Then, 	
	\begin{equation*}
	\begin{aligned}
	\mathbb{E}_{\mathbb{Q}} \Bigg(  max  \bigg\{   \sum_{i=1}^{N} \omega^{i} \tilde{T}_t^i  - K \,,\,0 \bigg\} dx \biggr\rvert  \mathcal{F}_t \Bigg) &= \int_{\Lambda^{1}(t,x)}^{+\infty} \bigg( \psi(t,x) + \big(\xi(t,x) + 2 \varDelta(t,x)  \big)^{\frac{1}{2}}z - 
	\\& \psi(t,x) - \Lambda^{1}(t,x) \big(\xi(t,x) + 2 \varDelta(t,x)  \big)^{\frac{1}{2}}\bigg)\dfrac{e^{-\frac{1}{2}z^2}}{\sqrt{2\pi}} dz
	\\
	&=  \int_{\Lambda^{1}(t,x)}^{+\infty} \bigg(\big(\xi(t,x) + 2 \varDelta(t,x)  \big)^{\frac{1}{2}}z - \\ &\Lambda^{1}(t,x) \Big(\xi(t,x) + 2 \varDelta(t,x)  \Big)^{\frac{1}{2}}\bigg) \dfrac{e^{-\frac{1}{2}z^2}}{\sqrt{2\pi}} dz
	\\
	&= \Big(\xi(t,x) + 2 \varDelta(t,x)  \Big)^{\frac{1}{2}}
	\bigg( \int_{\Lambda^{1}(t,x)}^{+\infty} \dfrac{ze^{-\frac{1}{2}z^2}}{\sqrt{2\pi}} dz + \Lambda^{1}(t,x)\Phi\big(-\Lambda^{1}(t,x)
	\bigg)
	\\
	&= \big(\xi(t,x) + 2 \varDelta(t,x)\big)^{\frac{1}{2}}
	\bigg(\dfrac{e^{-\frac{1}{2}\Lambda(t,x)^2}}{\sqrt{2\pi}} + \Lambda(t,x)\Phi\big(\Lambda(t,x)\bigg)
	\\
	&=\big( \xi (t,x) + 2\varDelta(t,x) \big)^{\ \frac{1}{2}} \big(\phi \big(\Lambda(t,x)\big) + \Lambda(t,x)\Phi\big(\Lambda(t,x)\big)\big).
	\end{aligned}
	\end{equation*}
	Therefore,
	\begin{equation*}
	F_{GDD}^{N}(t,t_1,t_2; D)  = \int_{t_1}^{t_2} \big( \xi (t,x) + 2\varDelta(t,x) \big)^{\ \frac{1}{2}} \big(\phi \big(\Lambda(t,x)\big) + \Lambda(t,x)\Phi\big(\Lambda(t,x)\big)\big)dx,
	\end{equation*} 
	where 
	\begin{equation*}
	\Lambda(t,x) = - \Lambda^1(t,x) = \dfrac{\psi (t,x)- C}{ \bigg(\xi(t,x) + 2 \varDelta(t,x)  \bigg)^{\frac{1}{2}}}.
	\end{equation*}
\end{proof}

\section{Discussion and Conclusion}

The agricultural sector employs a large workforce in Ghana and serves as the principal source of income for most people, especially small-holder farmers in the Northern savanna. This sector is vulnerable to climate shocks and, hence, there is a need for a reliable and efficient insurance product (weather derivative) for small-holder farmers and stakeholders. However, most farmers are unwilling to buy this product as a result of high basis risks in the product design and pricing. To reduce basis risk in weather derivative design and pricing, in this study, the historical relationship between maize yield and some selected weather variables were determined by using machine learning ensemble technique and feature importance. Maize yield was chosen as the proxy for crop yields, due to its economic importance to the farmers in the Northern savanna. The feature importance gave a score of the importance of each weather variable in building the ensemble-learning model. The results indicated that average temperature was the most important weather variable that affected the yield of maize. Consequently, average daily temperature was used as the underlying weather variable for weather derivative design. In this way, product design basis risk was mitigated.

A time-varying daily average temperature model was proposed. The model was a mean-reverting process with time-varying speed of mean-reversion, seasonal mean, and a local volatility that captured the local variations of the daily average temperature. The model was extended to a multi-dimensional model for different but correlated locations. For analytical tractability for the pricing models, the residuals of the daily average temperature were assumed to be normally distributed. Our models captured most of the stylized facts of temperature, such as mean-reversion, seasonality, volatility, and locality features of the selected locations. Using the proposed average temmodels, closed-form pricing formulas for futures, options on futures, basket futures for cumulative average temperature (CAT), and growing degree-days (GDD) were presented. Since there is not yet a real WD market for the selected locations, we assume a constant market price of risk in the pricing models.

With these efficient and reliable pricing models, basis risks will be mitigated. As a result, there will be an increase in the willingness to pay for the contracts on the farmers’ side and trading activities in the WD market will also increase. Using the proposed spatial-temporal pricing model, it will be cost efficient to buy contracts for different but correlated farming locations rather than a single farming location. Farmers and other agricultural stakeholders can control covariate risks and hedge their crops against the perils of weather uncertainties.

In summary, basis risks (product-design and geographical) can mitigated if weather derivatives are properly designed.

\section*{Acknowledgements}
The first author wishes to thank African Union and Pan African University, Instutute for Basic Sciences Technology and Innovation,Kenya for their financial support for this research.
\section*{Disclosure statement}
The authors declare that there is no conflict of interest regarding the publication of this paper

%%%%%%%%%%%%%%%%%%%%%%% Acknowledgement %%%%%%%%%%%%%%%%%

\bibliographystyle{apa}
\bibliography{reference}

\end{document}